\documentclass[journal,twoside,web]{ieeecolor}
\usepackage{generic}
\usepackage{cite}
\usepackage{amsmath,amssymb,amsfonts}
\usepackage{algorithmic}
\usepackage{graphicx}
\usepackage{algorithm,algorithmic}
\usepackage{hyperref}
\hypersetup{hidelinks=true}
\usepackage{textcomp}

\newtheorem{assumption}{Assumption} 
\newtheorem{lemma}{Lemma}
\newtheorem{theorem}{Theorem}
\DeclareMathOperator*{\argmax}{arg\,max}
\DeclareMathOperator*{\argmin}{arg\,min}
\usepackage{booktabs}
\makeatletter
\@webfalse
\def\ps@titlepagestyle{%
  \def\@oddhead{}%
  \def\@evenhead{}%
  \def\@oddfoot{\hfil\thepage\hfil}%
  \def\@evenfoot{\hfil\thepage\hfil}%
}
\def\ps@headings{%
  \def\@oddhead{}%
  \def\@evenhead{}%
  \def\@oddfoot{\hfil\thepage\hfil}%
  \def\@evenfoot{\hfil\thepage\hfil}%
}
\makeatother
\makeatletter

\def\section{\@startsection{section}{1}{\z@}{3.0ex plus 1.5ex minus 1.5ex}%
{0.7ex plus 1ex minus 0ex}{\color{black}\centering
   \fontencoding{T1}%
   \fontfamily{phv}%
   \fontseries{m}%
   \fontshape{n}%
   \fontsize{10}{12}%
   \selectfont
   \scshape}}%

\makeatother

\begin{document}
\title{Event-triggered Dual Gradient Tracking for Distributed Resource Allocation}
\author{Xiayan Xu, Xiaomeng Chen, Dawei Shi, and Ling Shi
\thanks{Xiayan Xu, Xiaomeng Chen and Ling Shi are with the Department of Electronic and Computer Engineering, the Hong Kong University of Science and Technology, Hong Kong, China. Ling Shi is also with the Department of Chemical and Biological Engineering, the Hong Kong University of Science and Technology, Hong Kong, China. (e-mail: xxucj@connect.ust.hk;~eexchen@ust,hk;~eesling@ust.hk). }
\thanks{Dawei Shi is with School of Automation, Beijing Institute of Technology, Beijing 100081, China (e-mail: daweishi@bit.edu.cn).}
\thanks{The work by X. Xu, X. Chen, and L. Shi is supported by the Hong Kong RGC Research Fund CRS\textunderscore HKUST601/22.}
}
\maketitle

\begin{abstract}
High communication costs create a major bottleneck for distributed resource allocation over unbalanced directed networks. Conventional dual gradient tracking methods, while effective for problems on unbalanced digraphs, rely on periodic communication that creates significant overhead in resource-constrained networks. This paper introduces a novel event-triggered dual gradient tracking algorithm to mitigate this limitation, wherein agents communicate only when local state deviations surpass a predefined threshold. We establish comprehensive convergence guarantees for this approach. First, we prove sublinear convergence for non-convex dual objectives and linear convergence under the Polyak-Łojasiewicz condition. Building on this, we demonstrate that the proposed algorithm achieves sublinear convergence for general strongly convex cost functions and linear convergence for those that are also Lipschitz-smooth. Numerical experiments confirm that our event-triggered method significantly reduces communication events compared to periodic schemes while preserving comparable convergence performance.
\end{abstract}

\begin{IEEEkeywords}
Directed graphs, Distributed resource allocation, Event-triggered scheme, Polyak-Łojasiewicz condition
\end{IEEEkeywords}

\section{Introduction}
\label{sec:introduction}
\IEEEPARstart{D}{istributed} resource allocation (DRA) provides a mathematical framework for agents in a networked system to collaboratively optimize shared resources distribution under global and local constraints. This paradigm is critical to large-scale systems like the economic dispatch of power in smart grids~\cite{zhang2012convergence}, bandwidth allocation in wireless sensor networks~\cite{xiao2004simultaneous}, and cooperative task allocation in robotic swarms~\cite{camisa2022multi}. In these scenarios, agents exchange local information over a communication network to collectively find optimal solutions that satisfy system-wide constraints.

\subsection{Related Work}
Developing efficient DRA algorithms has been a research focus for decades. A central issue is designing methods that operate over realistic networks with minimal communication overhead. Existing research has thus evolved along two main axes: robustness to complex network topologies and communication efficiency.

The first major challenge pertains to the underlying network topology. Discrete-time DRA algorithms are broadly categorized into primal~\cite{lakshmanan2008decentralized, nedic2018improved,wu2021new} and dual methods~\cite{xu2017distributed,zhang2012convergence, falsone2020tracking, jiang2022distributed, yang2016distributed, zhang2020distributed}. While primal (sub)gradient-based methods can converge faster according to empirical studies~\cite{nedic2018improved}, they typically require undirected or weight-balanced graphs, which is impractical for systems with asymmetric links. Dual methods, however, are well-suited for such topologies as they can transform the primal DRA problem into a dual distributed optimization (DO) problem and then leverage advanced consensus-based DO algorithms developed for digraphs~\cite{yang2019survey}.

Several distinct strategies have emerged within the dual framework to handle network asymmetry. Although Xu et al.~\cite{xu2017distributed} proposed the non-negative surplus-based averaging (NN-SURPLUS) algorithm for time-varying unbalanced digraphs, empirical analysis reveals slow convergence and sensitivity to initialization and parameter settings. Another prominent strategy is based on the alternate direction method of multipliers (ADMM)~\cite{falsone2020tracking, jiang2022distributed} which exhibits strong empirical performance~\cite{boyd2011distributed}. Researchers have made significant progress in extending ADMM to directed graphs. For example, the D-ADMM-FTERC algorithm~\cite{jiang2022distributed} achieves a provable $\mathcal{O}(1/k)$ convergence rate for convex and non-differentiable objectives over directed topologies, but its reliance on global network information and heavy communication for exact consensus limits scalability.

An alternative major approach and central focus of our research is gradient tracking methods~\cite{yang2016distributed, zhang2020distributed}. Early approaches addressed asymmetry using push-sum protocols~\cite{nedic2014distributed}; for instance, Yang et al.~\cite{yang2016distributed} introduced push-sum gradient techniques for strictly convex objectives over digraphs, though they require asymptotic eigenvector computation of the weight matrix. This line of research led to foundational algorithms like Push-DIGing~\cite{nedic2017achieving}, ExtraPush~\cite{zeng2017extrapush}, and DEXTRA~\cite{xi2017dextra}, which all established theoretical convergence guarantees for DO over digraphs. More recently, the distributed dual gradient tracking (DDGT) algorithm~\cite{zhang2020distributed}, based on push-pull gradients (PPG)~\cite{pu2020push}, achieved linear convergence for strongly convex and Lipschitz smooth objectives, and sublinear for general strongly convex problems over unbalanced digraphs, without requiring global information.

The second challenge is communication efficiency. The aforementioned algorithms are typically periodic and time-triggered, incurring prohibitive overhead in resource-constrained networks. Event-triggered (ET) mechanisms address this by restricting agents to communicate only when local state deviations exceed predefined thresholds. The ``communicate-when-necessary'' paradigm effectively reduces network usage while maintaining performance across applications such as remote state estimation~\cite{wu2012event}, decentralized control~\cite{deng2021distributed}, and parallel machine learning~\cite{solowjow2020event}, motivating its integration into DO~\cite{huang2024distributed, chen2024efficient, liu2019distributed} and DRA~\cite{wan2021adaptive, li2022distributed, guo2022distributed}. Dai et al.~\cite{dai2020distributed} proposed a discrete-time ET algorithm for DRA over balanced digraphs, requiring careful initialization and assuming no local constraints. Li et al.~\cite{li2022distributed} later introduced an initialization-free Mirror-P-EXTRA algorithm for undirected graphs with competitive convergence rates~\cite{nedic2018improved}. To improve adaptability and enlarge triggering intervals, dynamic ET mechanisms have been incorporated into discrete-time DRA algorithms~\cite{dong2022distributed, yuan2023distributed, chen2025distributed}, though all these methods assume undirected or balanced networks. For unbalanced digraphs, Dong et al.~\cite{dong2024distributed} proposed an asynchronous event-triggered push-pull gradient tracking method (DAAET-PPG), which focuses on asynchrony---a distinct research direction---and is limited to DO problems with strongly convex and Lipschitz smooth objectives. To date, no event-triggered discrete-time DRA solution with rigorous convergence guarantees has been established for unbalanced networks.

The current frontier of DRA research lies in developing communication-efficient algorithms with provable convergence over general unbalanced digraphs. While recent efforts have touched on this intersection, they either adopt different algorithmic frameworks or, like DAAET-PPG~\cite{dong2024distributed}, are restricted to strongly convex dual objectives. This highlights a critical gap, since strong convexity in the primal DRA problem does not guarantee the same in the dual DO problem. This necessitates a more general theoretical foundation for the convergence of ET-PPG algorithms.

\subsection{Contributions}
This paper develops a novel event-triggered distributed dual gradient tracking (ET-DGT) framework for DRA problems over unbalanced digraphs. We present a comprehensive theoretical analysis, establishing convergence rates that match state-of-the-art periodic algorithms while significantly reducing communication overhead. Our main contributions are threefold:
\begin{enumerate}
	\item \textbf{Communication-efficient algorithm for DRA on unbalanced digraphs:} We propose a novel event-triggered dual gradient tracking algorithm (ET-DGT; Algorithm~\ref{alg:etdgt}) for solving DRA problems with strongly convex cost functions over unbalanced digraphs. By integrating a state-dependent triggering mechanism with a push-pull consensus protocol, agents asynchronously exchange information only when local updates are significant. This design substantially reduces communication compared to periodic methods such as DDGT~\cite{zhang2020distributed} and NN-SURPLUS~\cite{xu2017distributed}.\\
	
	\item \textbf{Theoretical foundation for ET-PPG with weaker assumptions:} We establish a general theoretical foundation for event-triggered push-pull gradient (ET-PPG) methods under weaker assumptions. Specifically, we provide sufficient conditions for sublinear convergence with general non-convex objectives (Theorem~\ref{thm1:ET-PPG_nonconvex}), and linear convergence under the Polyak-Łojasiewicz condition (Theorem~\ref{thm2:ET-PPG_PL}). These results extend ET-PPG beyond the restrictive strong convexity setting, addressing a key theoretical gap in recent non-convex push-pull methods.\\
	
	\item \textbf{Rigorous convergence rate guarantees for ET-DGT:} We provide a rigorous analysis of the ET-DGT algorithm, deriving sufficient conditions on the constant step-size to ensure primal convergence (Theorem~\ref{thm3:primal conv}). Specifically, we prove sublinear $\mathcal{O}(1/k)$ convergence for strongly convex objectives, and linear $\mathcal{O}(\lambda^k)$ convergence with $0<\lambda<1$ for strongly convex and Lipschitz smooth objectives (Theorem~\ref{thm4:convrate}). These results demonstrate that ET-DGT matches the convergence rates of state-of-the-art periodic algorithms~\cite{zhang2020distributed}, while significantly reducing communication overhead.
	
\end{enumerate}

\subsection{Organization of the Paper}
The remainder of this article is organized as follows. Section~\ref{sec:problemstate} formulates the event-triggered DRA problem over unbalanced directed graphs. Section~\ref{sec:ETDGT algorithm} introduces the proposed ET-DGT algorithm. Section~\ref{sec:convergence} presents the comprehensive convergence analysis, establishing both sublinear and linear rates under suitable conditions. Section~\ref{sec:simulation} reports numerical experiments on economic dispatch problems, demonstrating the algorithm’s communication efficiency and convergence performance. Section~\ref{sec:conclusion} concludes with a summary of contributions and future research directions.

\textit{Notations:}
Let Italic lowercase letters $x$, bold lowercase letters $\mathbf{x}$ and bold uppercase letters $\mathbf{X}$ denote scalar, vector and matrix variables, respectively. Let $\mathbf{1} \in \mathbb{R}^{n}$ and $\mathbf{0} \in \mathbb{R}^{n}$ denote all-one and all-zero vector, respectively. Let $\left \| \cdot \right \|$ denote the $l_2$-norm of any vector and induced 2-norm of any matrix. Let $\succeq$($\preceq$) denote element-wise inequality between matrices. For the communication topology, let $\mathcal{G}(\mathcal{N},\mathcal{E} )$ denote a directed graph where $\mathcal{N}=\{1,\ldots, n\}$ is the set of nodes and $\mathcal{E} \subseteq \mathcal{N} \times \mathcal{N} $ is the set of directed edges. An edge $(i,j)$ indicates that node $i$ can receive information from node $j$. A directed graph induced by a non-negative matrix $\mathbf{A}=[a_{ij}]\in \mathbb{R}^{n \times n}$ is denoted by $\mathcal{G}_{\mathbf{A}}(\mathcal{N}, \mathcal{E}_{\mathbf{A}})$ where $(i,j) \in \mathcal{E}_{\mathbf{A}}$ if and only if $a_{ij} > 0$. For each node $i \in \mathcal{N}$, let $\mathcal{N}_i= \{j \mid   a_{ij} >0\}$ denote its in-neighbour set, i.e., nodes that transmit information to node $i$. Similarly, let $\bar{\mathcal{N}}_i = \{j \mid a_{ji} > 0\}$ denote its out-neighbor set, i.e., nodes that receive information from $i$.

\section{Problem Statement}
\label{sec:problemstate}
This section formulates the distributed resource allocation (DRA) problem over unbalanced directed graphs. We begin with the primal formulation and its dual counterpart, which takes the form of a consensus-based distributed optimization (DO) problem. To solve the dual efficiently over unbalanced digraphs, we introduce the push-pull gradient tracking algorithm. An analysis of agent broadcast behavior reveals communication inefficiencies, motivating the development of our event-triggered solution.
\subsection{Distributed Resource Allocation Problem}
We consider a distributed resource allocation problem with $n$ agents communicating over an unbalanced digraph $\mathcal{G}$. The global objective of all agents is to solve the following DRA problem in a distributed manner:

\begin{equation}\label{p:primal}
	\begin{aligned}
		\min_{\mathbf{w}_1, \ldots, \mathbf{w}_n \in \mathbb{R}^{m}} \quad & \sum_{i=1}^{n} F_i(\mathbf{w}_i) \\
		\text{s.t.} \quad \quad \quad & \sum_{i=1}^{n} \mathbf{w}_i = \sum_{i=1}^{n} \mathbf{d}_i, \quad \mathbf{w}_i \in \mathcal{W}_i, \quad \forall i \in \mathcal{N}
	\end{aligned}
\end{equation}
where $\mathbf{w}_i\in \mathbb{R}^m$ is the local decision variable of agent $i$, $\mathcal{W}_i \in \mathbb{R}^{m}$ is the local closed and convex feasibility constraint set, $\mathbf{d}_i \in \mathbb{R}^m$ is the local demand of resource, and $F_i:\mathbb{R}^m \to \mathbb{R}$ is the local cost function. The overall resource supply constraint is $\sum_{i=1}^n \mathbf{d}_i$. Agents in the network are globally coupled by the resource supply-demand balance constraint.

In this paper, we make the following assumptions:
\begin{assumption}[Strong convexity and Slater's condition]\label{ass1:primal_cvx}
	\leavevmode
	\begin{enumerate}
		\item Each local cost function $F_i$ is $\frac{1}{L}$-strongly convex, i.e., $\forall i \in \mathcal{N}$ and $\forall \mathbf{w},\mathbf{w}' \in \mathbb{R}^m$, $ \left \| \nabla F_i(\mathbf{w})-\nabla F_i(\mathbf{w}') \right \| \ge \frac{1}{L}\left \| \mathbf{w}-\mathbf{w}' \right \|$.
		\item There exists at least one point in the relative interior $\mathcal{W}$ that satisfies the supply-demand balance constraint $\sum_{i=1}^{n} \mathbf{w}_i = \sum_{i=1}^{n} \mathbf{d}_i$, where $\mathcal{W}=\mathcal{W}_1 \times \cdots \times \mathcal{W}_n$.
	\end{enumerate}
\end{assumption}
Assumption~\ref{ass1:primal_cvx} ensures the existence of a unique optimal solution to problem~\eqref{p:primal} and establishes strong duality between the primal and dual formulations. This allows the coupling constraint to be effectively handled through dual optimization.

\subsection{Dual Problem of DRA}
To solve the globally constrained DRA problem~(\ref{p:primal}), we start by forming the dual problem and transforming it into a consensus-based DO problem. Define the Lagrangian $\mathcal{L}$ as follows:
\begin{align}\label{p:lag}
	\mathcal{L}(\mathbf{W},\mathbf{x}) = \sum_{i=1}^{n} F_i(\mathbf{w}_i) + \mathbf{x}^{\top}\left(\sum_{i=1}^{n} \mathbf{w}_i -  \sum_{i=1}^{n}\mathbf{d}_i\right),
\end{align}
where $\mathbf{W} = [\mathbf{w}_1, \ldots, \mathbf{w}_n]^{\top}\in \mathbb{R}^{n \times m}$ and $\mathbf{x}\in \mathbb{R}^m$ represents the Lagrange multiplier vector. Then the dual function is:
\begin{align*}
	\begin{aligned}
		g(\mathbf{x}) = &\inf_{\mathbf{W} \in \mathcal{W}} \mathcal{L}(\mathbf{W},\mathbf{x})\\
		= & \sum_{i=1}^{n} \inf_{\mathbf{w}_i \in \mathcal{W}_i} \left\{F_i(\mathbf{w}_i) + \mathbf{x}^{\top} \mathbf{w}_i \right\} -\sum_{i=1}^{n}\mathbf{x}^{\top}\mathbf{d}_i \\
		= & \sum_{i=1}^{n} -\left( F^*_i(-\mathbf{x}) + \mathbf{x}^{\top}\mathbf{d}_i \right),
	\end{aligned}
\end{align*}
where $F^*_i$ is the convex conjugate function for the pair $(F_i,\mathcal{W}_i)$~\cite{bertsekas1997nonlinear}:
\begin{equation} \label{eq:conj_Fi}
	F^*_i(\mathbf{x}) = \sup_{\mathbf{w}_i \in \mathcal{W}_i}\left\{ \mathbf{x}^{\top}\mathbf{w}_i - F_i(\mathbf{w}_i) \right \}.
\end{equation}
Thus the dual problem of~(\ref{p:primal}) is:
\begin{align*}
	\begin{aligned}
		\max_{\mathbf{x} \in \mathbb{R}^{m}} g(\mathbf{x}),
	\end{aligned}	
\end{align*}
which can be reformulated as a consensus-based DO problem:
\begin{align}\label{p:dual}
	\min_{\mathbf{x} \in \mathbb{R}^{m}} f(\mathbf{x}) = \sum_{i=1}^{n} f_i(\mathbf{x}),
\end{align}
where $f_i(\mathbf{x}) = F^*_i(-\mathbf{x}) + \mathbf{x}^\top\mathbf{d}_i$ is agent $i$'s dual objective.

Based on the Fenchel duality between strong convexity and Lipschitz continuous gradient~\cite{zhou2018fenchel}, the $\frac{1}{L}$-strong convexity of $F_i$ implies the $L$-Lipschitz smoothness of $F^*_i$ and the available supremum in~\eqref{eq:conj_Fi}. According to Danskin's theorem~\cite{bertsekas1997nonlinear}, the gradient of $F_i^*$ is:
\begin{align}\label{eq:grad_Fi*}
	\nabla F^*_i(\mathbf{x}) = \argmax_{\mathbf{w}_i \in \mathcal{W}_i} \left\{ \mathbf{x}^{\top}\mathbf{w}_i-F_i(\mathbf{w}_i) \right\},
\end{align}
from which we derive the dual gradient as:
\begin{align}\label{eq:grad_fi}
	\begin{aligned}
		\nabla f_i(\mathbf{x})= -\argmin_{\mathbf{w}_i \in \mathcal{W}_i} \left\{ \mathbf{x}^{\top}\mathbf{w}_i+F_i(\mathbf{w}_i) \right\} + \mathbf{d}_i.
	\end{aligned}
\end{align}
Assumption~\ref{ass1:primal_cvx} ensures strong duality and primal solution uniqueness, i.e., $F^* = g^*$. By the equivalent reformulation of dual problem~\eqref{p:dual}, this implies $F^* = -f^*$ and establishes the Fenchel-Young equality between optimal variables, i.e., $F_i(\mathbf{W}^*_i) + F^*_i(-\mathbf{x}^*) = -(\mathbf{x}^{*})^{\top}\mathbf{W}^*_i$. This equality enables the recovery of optimal primal variable $\mathbf{W}^*_i$ from optimal dual variable $\mathbf{x}$. Therefore, solving the dual problem~(\ref{p:dual}) is equivalent to solving the primal one~(\ref{p:primal}) and we focus on distributed algorithms for the dual formulation.

\subsection{Distributed Dual Gradient Tracking Algorithm}
We consider a directed communication graph $\mathcal{G}$ and adopt the push-pull gradient (PPG) algorithm, which has demonstrated effectiveness for DO in unbalanced networks. The distributed dual gradient tracking (DDGT) algorithm~\cite{zhang2020distributed} applies PPG to solve the dual problem~(\ref{p:dual}):
\begin{subequations}
	\begin{align}
		\tilde{\mathbf{w}}_{i,k+1}&=\sum_{j=1}^{n} r_{ij}\tilde{\mathbf{w}}_{j,k}+\alpha \mathbf{s}_{i,k}, \label{eq:std_w}\\
		\mathbf{w}_{i,k+1}&=\argmin_{\mathbf{w} \in \mathcal{W}_i}\{F_i(\mathbf{w}) - \tilde{\mathbf{w}}_{i,k+1}^{\top}\mathbf{w}\},\label{eq:std_ww}\\
		\mathbf{s}_{i,k+1}&=\sum_{j=1}^{n}c_{ij}\mathbf{s}_{j,k} - (\mathbf{w}_{i,k+1}-\mathbf{w}_{i,k}),\label{eq:std_s}
	\end{align}	
\end{subequations}
where $\alpha > 0$ is a constant step size, $\tilde{\mathbf{w}}_{i,k}$ is agent $i$'s local estimate of the dual variable $\mathbf{x}$, $\mathbf{s}_{i,k}$ is its local estimate of the global dual gradient $\nabla f(\mathbf{x})$, and $\mathbf{w}_{i,k}$ is the local optimal primal variable obtained by solving the subproblem in~(\ref{eq:std_ww}).  Agents exchange $\tilde{\mathbf{w}}_{i,k}$ and $\mathbf{s}_{i,k}$ through two distinct directed graphs: $\mathcal{G}_{\mathbf{R}}$ for $\tilde{\mathbf{w}}$ and $\mathcal{G}_{\mathbf{C}^{\top}}$ for $\mathbf{s}$. Specifically, $\mathcal{G}_{\mathbf{R}}$ (``pull'' network) is induced by matrix $\mathbf{R} = [r_{ij}]\in \mathbb{R}^{n \times n}$, where agent $i$ receives $\tilde{\mathbf{w}}_{j,k}$ from in-neighbors $j \in \mathcal{N}_{\mathbf{R},i}$. Meanwhile, $\mathcal{G}_{\mathbf{C}^{\top}}$ (``push'' network) is induced by $\mathbf{C}^{\top} = [c_{ij}]\in \mathbb{R}^{n \times n}$, where agent $j$ broadcasts $c_{ij}\mathbf{s}_{j,k}$ to out-neighbors $i \in \bar{\mathcal{N}}_{\mathbf{C},j}$. The following assumptions of the communication graphs are made throughout the paper.

\begin{assumption}\label{ass2:graph}
	The graphs $\mathcal{G}_{\mathbf{R}}$ and $\mathcal{G}_{\mathbf{C}^{\top}}$ each contains at least one spanning tree. Moreover, there exists at least one node that is a root of a spanning tree for both $\mathcal{G}_{\mathbf{R}}$ and $\mathcal{G}_{\mathbf{C}^{\top}}$.
\end{assumption}

\begin{assumption}\label{ass3:graph_weight}
	The matrix $\mathbf{R}$ is row-stochastic and $\mathbf{C}$ is column-stochastic, i.e., $\mathbf{R1}=\mathbf{1}$ and $\mathbf{1}^{\top}\mathbf{C}=\mathbf{1}^{\top}$.
\end{assumption}

Notice that Assumption~\ref{ass2:graph} is standard in the literature~\cite{pu2020push}. It relaxes the strongly-connected graph requirement of prior works~\cite{xi2017dextra} and enables less restrictive network design of $\mathcal{G}_{\mathbf{R}}$ and $\mathcal{G}_{\mathbf{C}^{\top}}$. Assumption~\ref{ass3:graph_weight} further ensures that $\mathbf{R}$ and $\mathbf{C}$ preserve the necessary averaging properties in the push--pull dynamics, and such an asymmetric $\mathbf{R}$--$\mathbf{C}$ structure has been widely adopted in the literature to guarantee convergence under directed and unbalanced networks~\cite{pu2020push}.

By leveraging the push-pull communication scheme, the DDGT algorithm effectively tracks both the global dual variable and its gradient, ensuring convergence of the primal variable and consensus error. However, maintaining dual communication graphs introduces considerable communication overhead. To address this issue, Section~III introduces an event-triggered mechanism to reduce communication costs.

\section{Event-triggered Algorithm Development}
\label{sec:ETDGT algorithm}
This section introduces a deterministic event-triggered mechanism for transmitting $\tilde{\mathbf{w}}_{i}$ and $\mathbf{s}_{i}$ in the DDGT algorithm. We first present the triggering mechanisms, then reformulate ET-DGT to align with the PPG framework for subsequent convergence analysis.

\subsection{Event-triggered Mechanisms}
Event-triggered mechanisms reduce communication overhead in distributed optimization~\cite{yang2019survey} by enabling aperiodic information exchange, where communication occurs only when local updates violate predefined triggering conditions. To reduce communication frequency in DDGT, we design separate event-triggered mechanisms for $\tilde{\mathbf{w}}_i$ and $\mathbf{s}_i$.

Let $k_{\tilde{\mathbf{w}}_{i}}^{\ell}$ denote the $\ell$-th triggering instant for $\tilde{\mathbf{w}}_{i}$, $\ell \in \mathbb{Z}^+$, and $k_{\mathbf{s}_i}^{\pi}$ the $\pi$-th triggering instant for $\mathbf{s}_i$, $\pi \in \mathbb{Z}^+$. Define the broadcasted variables as $\hat{\mathbf{w}}_i$ and $\hat{\mathbf{s}}_i$ respectively:
\begin{align*}
	\hat{\mathbf{w}}_{i, k} &= \begin{cases}
		\tilde{\mathbf{w}}_{i,k},  &\text{ if } k\in\kappa_{\tilde{\mathbf{w}}_i} \\
		\hat{\mathbf{w}}_{i,k-1}, &\text{ otherwise } 
	\end{cases},	\\
	\hat{\mathbf{s}}_{i,k}  &= \begin{cases}
		\mathbf{s}_{i,k},  &~\text{ if } k\in\kappa_{\mathbf{s}_i} \\
		\hat{\mathbf{s}}_{i,k-1}, &~\text{ otherwise }
	\end{cases}	,		
\end{align*}
where $\kappa_{\tilde{\mathbf{w}}_i}:=\{k_{\tilde{\mathbf{w}}_{i}}^{0}, k_{\tilde{\mathbf{w}}_{i}}^{1},k_{\tilde{\mathbf{w}}_{i}}^{2},\ldots\}$ and $\kappa_{\mathbf{s}_i}:=\{k_{\mathbf{s}_i}^{0}, k_{\mathbf{s}_i}^{1},k_{\mathbf{s}_i}^{2},\ldots\}$ denote the collections of all the triggering instances for each agent $i$. The $(\ell+1)$-th triggering instance for $\tilde{\mathbf{w}}_i$ is derived from the following triggering law:
\begin{align}\label{eq:etlaw1}
	k_{\tilde{\mathbf{w}}_{i}}^{\ell+1} := \inf \left \{k: k > k_{\tilde{\mathbf{w}}_{i}}^{\ell} , \big\| \tilde{\mathbf{w}}_{i,k} - \hat{\mathbf{w}}_{i,k_{\tilde{\mathbf{w}}_{i}}^{\ell}} \big\| \geq e_{\tilde{\mathbf{w}}_i,k} \right\},
\end{align}
where $\{e_{\tilde{\mathbf{w}}_i,k}\}_{k=0}^{\infty}$ is the deterministic triggering threshold that needs to be designed later. Similarly, the $(\pi+1)$-th triggering instance for $\mathbf{s}_i$ is derived from:
\begin{align}\label{eq:etlaw2}
	k_{\mathbf{s}_i}^{\pi +1} := \inf \left \{k: k > k_{\mathbf{s}_i}^{\pi}, \big\| \mathbf{s}_{i,k} - \hat{\mathbf{s}}_{i,k_{\mathbf{s}_i}^{\pi}} \big\| \ge e_{\mathbf{s}_i,k} \right \},
\end{align}
where $\{e_{\mathbf{s}_i,k}\}_{k=0}^{\infty}$ is also the triggering threshold. For the triggering thresholds we make the following assumption:
\begin{assumption}\label{ass4:trigger}
	Define $e_k=\max_{i \in \mathcal{N}} (e_{\tilde{\mathbf{w}}_i,k}, e_{\mathbf{s}_i,k}), \forall k \in \mathbb{Z}^+$. The sequence $\{e_k\}_{k=0}^{\infty}$ is non-increasing and summable, i.e., $S_e=\sum_{k=0}^{\infty}e_k<\infty$.
\end{assumption}

\subsection{Event-triggered Dual Gradient Tracking Algorithm}
Building on the event-triggered mechanisms from previous subsection, we propose the event-triggered dual gradient tracking (ET-DGT) algorithm in Algorithm~\ref{alg:etdgt}. Each agent $i$ broadcasts $\tilde{\mathbf{w}}_{i,k}$ over $\mathcal{G}_\mathbf{R}$ and $c_{li}\mathbf{s}_{i,k}$ over $\mathcal{G}_{\mathbf{C}^{\top}}$ only when the triggering conditions in~\eqref{eq:etlaw1} and~\eqref{eq:etlaw2} are satisfied, respectively. Receivers use the broadcasted values $\hat{\mathbf{w}}_{i,k}$ and $c_{li}\hat{\mathbf{s}}_{i,k}$ rather than the true values for local updates. This mechanism introduces communication-induced errors, necessitating further convergence analysis to guarantee solving performance.

\begin{algorithm}
	\caption{Event-triggered Dual Gradient Tracking}
	\label{alg:etdgt}
	\begin{algorithmic}[1] 
		
		\STATE \textbf{Input:} Step size $\alpha$, event-triggering thresholds $\{e_{\mathbf{s}_i,k}\}$ and $\{e_{\tilde{\mathbf{w}}_i,k}\}$ for any $i \in \mathcal{N}$ and $k \in \mathbb{R}^+$.
		\STATE \textbf{Initialization:} Let variables $\tilde{\mathbf{w}}_{i,0} = \mathbf{0}$, ${\mathbf{w}}_{i,0} = \mathbf{0}$, and {${\mathbf{s}}_{i,0} = \mathbf{d}_i-\mathbf{w}_{i,0}$}. Let initial triggering instances $k_{\tilde{\mathbf{w}}_{i}}^{0} = 0$ and $k_{\mathbf{s}_i}^{0} = 0$.
		\FOR{$k = 0,1,\ldots,K$}
		
		\FOR{\text{each agent} $i \in \mathcal{N}$}
		\IF {$k \in \kappa_{\tilde{\mathbf{w}}_i}$}
		\STATE Agent $i$ broadcasts $\tilde{\mathbf{w}}_{i,k}$ to each agent $j \in \bar{\mathcal{N}}_{\mathbf{R},i}$.
		\ENDIF
		\IF {$k \in \kappa_{\mathbf{s}_i}$}
		\STATE Agent $i$ broadcasts $c_{li}{\mathbf{s}}_{i,k}$ to each agent $l \in \bar{\mathcal{N}}_{\mathbf{C},i}$.
		\ENDIF
		\ENDFOR
		
		\FOR{\text{each agent} $i \in \mathcal{N}$}
		\STATE Agent $i$ updates local variables:
		\begin{subequations}
			\begin{align}\label{eq:et_w}
				\tilde{\mathbf{w}}_{i,k+1}=&\tilde{\mathbf{w}}_{i,k} + \sum_{j=1}^n r_{ij}(\hat{\mathbf{w}}_{j,k} - \hat{\mathbf{w}}_{i,k})  +\alpha \mathbf{s}_{i,k} , \\ \label{eq:et_ww}
				\mathbf{w}_{i,k+1}=&\argmin_{\mathbf{w} \in \mathcal{W}_i}\{F_i(\mathbf{w}) - \tilde{\mathbf{w}}_{i,k+1}^{\top}\mathbf{w}\},\\
				s_{i,k+1} = &\mathbf{s}_{i,k} + \sum_{j=1}^{n}{c_{ij}}(\hat{\mathbf{s}}_{j,k} - \hat{\mathbf{s}}_{i,k}) \nonumber \\
				&- (\mathbf{w}_{i,k+1}-\mathbf{w}_{i,k}),\label{eq:et_s}
			\end{align}
		\end{subequations}	
	    \ENDFOR
		\ENDFOR

	\end{algorithmic}
\end{algorithm}
To address event-triggered information errors, define the transmission errors as: $\boldsymbol{\zeta}_{i,k} = \hat{\mathbf{w}}_{i,k} - \tilde{\mathbf{w}}_{i,k}$ and $\boldsymbol{\xi}_{i,k} = \hat{\mathbf{s}}_{i,k} - \mathbf{s}_{i,k}$.
Based on event-triggered mechanisms, the communication errors are bounded by the summable thresholds for all $i \in \mathcal{N}$ and $k\in \mathbb{Z}^+$, i.e., $\left \| \boldsymbol{\zeta}_{i,k} \right \| \le e_{\tilde{\mathbf{w}}_i,k}$ and $\left \| \boldsymbol{\xi}_{i,k} \right \| \le e_{\mathbf{s}_i,k}$.

The event-triggered local updates in~\eqref{eq:et_w}--\eqref{eq:et_s} preserve the push-pull structure of standard DDGT in~\eqref{eq:std_w}--\eqref{eq:std_s}, using last-broadcast values for consensus and gradient tracking. 
Define the stacked vectors: $\tilde{\mathbf{W}}_k = [\tilde{\mathbf{w}}_{1,k}, \ldots, \tilde{\mathbf{w}}_{n,k}]^{\top} \in \mathbb{R}^{n \times m}$, $\mathbf{S}_k = [\mathbf{s}_{1,k}, \ldots, \mathbf{s}_{n,k}]^{\top} \in \mathbb{R}^{n \times m}$, $\boldsymbol{\zeta}_k = [\boldsymbol{\zeta}_{1,k}, \ldots, \boldsymbol{\zeta}_{n,k}]^{\top} \in \mathbb{R}^{n \times m}$, and $\boldsymbol{\xi}_k = [\boldsymbol{\xi}_{1,k}, \ldots, \boldsymbol{\xi}_{n,k}]^{\top} \in \mathbb{R}^{n \times m}$. These yield a compact presentation of the event-triggered local updates:
\begin{subequations}
	\begin{align}
		\tilde{\mathbf{W}}_{k+1} &= \tilde{\mathbf{W}}_{k} + (\mathbf{R}-\mathbf{I}) (\tilde{\mathbf{W}}_{k} + \boldsymbol{\zeta}_k) + \alpha \mathbf{S}_k , \label{eq:et_big_w}\\
		\mathbf{S}_{k+1} &= \mathbf{S}_{k} + (\mathbf{C}-\mathbf{I})(\mathbf{S}_{k} + \boldsymbol{\xi}_k) - (\mathbf{W}_{k+1} - \mathbf{W}_{k}) \label{eq:et_big_s}
	\end{align}
\end{subequations}
Initializing $\mathbf{S}_0 = -(\mathbf{W}_0 - \mathbf{D})$, where $\mathbf{D} = [\mathbf{d}_1, \ldots, \mathbf{d}_n]^{\top} \in \mathbb{R}^{n\times m}$, we can derive by induction that
\begin{align}\label{eq:sum_sk}
	\begin{aligned}
		\mathbf{1}^{\top} \mathbf{S}_k = &\mathbf{1}^{\top} \mathbf{S}_0 - \mathbf{1}^{\top}(\mathbf{W}_{k} - 	\mathbf{W}_{0}) \\
		= & \mathbf{1}^{\top}(-\mathbf{W}_{k} + \mathbf{D}) \\
		= & \mathbf{1}^{\top}\nabla \mathbf{f}(\tilde{\mathbf{W}}_{k}),
	\end{aligned}
\end{align}
where the last equality follows from the gradient of dual objective function $\nabla f_i(x)$ in~\eqref{eq:grad_fi} and $\nabla \mathbf{f}(\tilde{\mathbf{W}}_{k}) = [\nabla f_1(\tilde{\mathbf{w}}_{1,k}), \ldots, \nabla f_n(\tilde{\mathbf{w}}_{n,k})]^{\top} \in \mathbb{R}^{n\times m}$. The telescoping sum in~\eqref{eq:sum_sk}, derived from the column-stochasticity of $\mathbf{C}$, reveals that $\mathbf{S}_k$ accurately tracks the global gradient despite the presence of event-triggered perturbations. Notably, the convergence of either $\mathbf{s}_k$ or $\nabla \mathbf{f}(\tilde{\mathbf{W}}_k)$ to zero ensures asymptotic satisfaction of the constraints by $\mathbf{W}_k$. These properties are fundamental to the subsequent convergence analysis.

\section{Convergence Analysis}
\label{sec:convergence}
This section analyzes the convergence of the proposed ET-DGT algorithm. The analysis begins by reformulating the ET-DGT dynamics \eqref{eq:et_big_w}--\eqref{eq:et_big_s} to align with the established framework of PPG methods, followed by the introduction of several key supporting lemmas. A primary challenge arises from the dual formulation: although the primal cost functions $F_i$ are strongly convex, the corresponding dual objectives $f_i$ may lack strong convexity. This necessitates a more general analytical approach. Accordingly, we first establish dual convergence without assuming convexity, and subsequently prove primal convergence even when the dual objectives are not strongly convex.

\subsection{Supporting Lemmas}
With a slight abuse of notation, let $\mathbf{Y}_k=\mathbf{S}_k$, $\mathbf{X}_k=-\tilde{\mathbf{W}}_k$, and substitute $\mathbf{W}_k = - \nabla \mathbf{f}(\mathbf{X}_k) + \mathbf{D}$ into \eqref{eq:et_big_w} and \eqref{eq:et_big_s}. The PPG-aligned compact form of ET-DGT algorithm is:
\begin{subequations}
	\begin{align}
		\mathbf{X}_{k+1} =& \mathbf{X}_k + (\mathbf{R}-\mathbf{I})(\mathbf{X}_k + \boldsymbol{\zeta}_k) - \alpha\mathbf{Y}_{k}, \label{eq:x} \\
		\mathbf{Y}_{k+1} =& \mathbf{Y}_k + (\mathbf{C}-\mathbf{I})(\mathbf{Y}_k + \boldsymbol{\xi}_k) \nonumber \\
		& + \nabla\mathbf{f} (\mathbf{X}_{k+1})-\nabla\mathbf{f} (\mathbf{X}_{k}). \label{eq:y}
	\end{align}
\end{subequations}
Under Assumption~\ref{ass2:graph}, we establish preliminary lemmas regarding the communication digraphs:
\begin{lemma}\label{lem1:perron}
	The row-stochastic matrix $\mathbf{R}$ admits a unique unit non-negative left eigenvector $\pi_{\mathbf{R}}$ w.r.t. eigenvalue $1$, i.e., $\pi_{\mathbf{R}}^{\top}\mathbf{R}=\pi_{\mathbf{R}}^{\top}$ and $\pi_{\mathbf{R}}^{\top}\mathbf{1}=1$. The column-stochastic matrix $\mathbf{C}$ admits a unique unit non-negative right eigenvector $\pi_{\mathbf{C}}$ w.r.t. eigenvalue $1$, i.e., $\mathbf{C}\pi_{\mathbf{C}}=\pi_{\mathbf{C}}$ and $\pi_{\mathbf{C}}^{\top}\mathbf{1}=1$.
	
\end{lemma}

\begin{lemma}\label{lem2:norm_bd}
	There exist matrix norms $\left \| \cdot \right \|_{\mathbf{R}}$ and $\left \| \cdot \right \|_{\mathbf{C}}$ such that $\sigma_{\mathbf{R}} := \left \| \mathbf{R}-\mathbf{1}\pi_{\mathbf{R}}^{\top} \right \|_{\mathbf{R}} < 1$ and $\sigma_{\mathbf{C}} := \left \| \mathbf{C}-\pi_{\mathbf{C}}\mathbf{1}^{\top} \right \|_{\mathbf{C}} <1$, where $\sigma_{\mathbf{R}}$ and $\sigma_{\mathbf{C}}$ can be chosen arbitrarily closed to the respective spectral radius of $\mathbf{R}-\mathbf{1}\pi_{\mathbf{R}}^{\top}$ and $\mathbf{C}-\pi_{\mathbf{C}}\mathbf{1}^{\top}$.
\end{lemma}

\begin{lemma}\label{lem3:matrix_norm}
	There exist constants $\delta_{\mathbf{R},\mathbf{C}}$ and $\delta_{\mathbf{C},\mathbf{R}}$ such that for any vector $\mathbf{x}$, we have $\left \| \cdot \right \|_{\mathbf{R}} \le \delta_{\mathbf{R},\mathbf{C}}\left \| \cdot \right \|_{\mathbf{C}}$, $\left \| \cdot \right \|_{\mathbf{C}} \le \delta_{\mathbf{C},\mathbf{R}}\left \| \cdot \right \|_{\mathbf{R}}$. Additionally, both norms satisfy $\left \| \cdot \right \|_{\mathbf{R}} \le \left \| \cdot \right \|_{2}$ and $\left \| \cdot \right \|_{\mathbf{C}} \le \left \| \cdot \right \|_{2}$ from the construction of norms $\left \| \cdot \right \|_{\mathbf{R}}$ and $\left \| \cdot \right \|_{\mathbf{C}}$.
\end{lemma}

\subsection{Convergence of Dual Problem}
The dual convergence analysis focuses on three key metrics: consensus error, gradient tracking error, and gradient norm. Define the following auxiliary variables: $\bar{\mathbf{x}}_k=\mathbf{X}_k^{\top}\pi_{\mathbf{R}}$, $\bar{\mathbf{y}}_k=\mathbf{Y}_k^{\top}\pi_{\mathbf{R}}$, and $\hat{\mathbf{y}}_k=\mathbf{Y}_k^{\top}\mathbf{1}=\nabla\mathbf{f} (\mathbf{X}_{k})^{\top}\mathbf{1}$. Then the metrics are characterized as: the consensus error in dual variable estimates $\left \|  \mathbf{X}_{k+1} - \mathbf{1} \bar{\mathbf{x}}_{k+1}^{\top} \right \|_{\mathbf{R}}$, the tracking error in dual gradient estimates $\left \| \mathbf{Y}_{k+1}-\pi_{\mathbf{C}}\hat{\mathbf{y}}_{k+1}^{\top} \right \|_{\mathbf{C}}$ and the gradient norm at the average dual variable $\left \|\nabla f (\bar{\mathbf{x}}_k)\right \|$. Before presenting the main convergence theorem, we establish several supporting lemmas.

\begin{lemma}\label{lem4:ineq_sys1}
	Under Assumptions~\ref{ass1:primal_cvx}--\ref{ass4:trigger}, the following linear system of inequalities w.r.t. $\left \|  \mathbf{X}_{k+1} - \mathbf{1} \bar{\mathbf{x}}_{k+1}^{\top} \right \|_{\mathbf{R}}$ and $\left \| \mathbf{Y}_{k+1}-\pi_{\mathbf{C}}\hat{\mathbf{y}}_{k+1}^{\top} \right \|_{\mathbf{C}}$ holds:
	\begin{align} \label{eq:inequal_sys_1}
		\begin{aligned}
			\begin{bmatrix}
				\left \| \mathbf{X}_{k+1}-\mathbf{1}\bar{\mathbf{x}}_{k+1}^{\top} \right \|_{\mathbf{R}} \\
				\left \| \mathbf{Y}_{k+1}-\pi_{\mathbf{C}}\hat{\mathbf{y}}_{k+1}^{\top} \right \|_{\mathbf{C}}
			\end{bmatrix}
			\preceq &\mathbf{P}
			\begin{bmatrix}
				\left \| \mathbf{X}_{k}-\mathbf{1}\bar{\mathbf{x}}_{k}^{\top} \right \|_{\mathbf{R}} \\
				\left \| \mathbf{Y}_{k}-\pi_{\mathbf{C}}\hat{\mathbf{y}}_{k}^{\top} \right \|_{\mathbf{C}}
			\end{bmatrix}\\
			&+ \mathbf{Q}
			\begin{bmatrix}
				\alpha \left \| \nabla f(\bar{\mathbf{x}}_k) \right \| \\
				\sqrt{n}e_k
			\end{bmatrix},
		\end{aligned}
	\end{align}
	where the inequality is taken element-wise, and the elements of matrices $\mathbf{P}$ and $\mathbf{Q}$ are given by:	
	\begin{align*}
		\mathbf{P} =& \begin{bmatrix}
			\sigma_{\mathbf{R}} + \alpha  \sqrt{n} L & \alpha \delta_{\mathbf{R},\mathbf{C}}\\
			\left( 2 + \alpha \sqrt{n} L \right) L \delta_{\mathbf{C},\mathbf{R}} & \sigma_{\mathbf{C}} + \alpha L
		\end{bmatrix}, \\
		\mathbf{Q} =& \begin{bmatrix}
			1 & 1+\sigma_{\mathbf{R}}\\
			L &  2L + 1 + \sigma_{\mathbf{C}}
		\end{bmatrix}.
	\end{align*}
\end{lemma}
\begin{proof}
	See APPENDIX A.
\end{proof}

Denote $\lambda \triangleq \rho(\mathbf{P})$ as the spectral radius of $\mathbf{P}$ and $\underline{\Psi}$ as the lower bound of $\Psi = \sqrt{(p_{11}-p_{12})^2+4p_{12}p_{21}}$. Based on these quantities, the positive constants$c_i, i=0,\ldots,5$ are introduced to simplify the expressions in subsequent lemmas:
\begin{align}
	c_0 &= \left \| \mathbf{Y}_{0}-\pi_{\mathbf{C}}\hat{\mathbf{y}}_{0}^{\top} \right \|_{\mathbf{C}} / \underline{\Psi} \le \sqrt{\sum_{i=1}^{n}||\nabla f_i(\mathbf{X}_0)||^2} / \underline{\Psi},\nonumber \\
	c_1 &= 1 + L / \underline{\Psi},\nonumber \\
	c_2 &= \left [1+\sigma_{\mathbf{R}}+ (2L+1+\sigma_{\mathbf{C}})/\underline{\Psi} \right ] \sqrt{n}, \nonumber \\
	c_3 &= \left \| \mathbf{Y}_{0}-\pi_{\mathbf{C}}\hat{\mathbf{y}}_{0}^{\top} \right \|_{\mathbf{C}}\le \sqrt{\sum_{i=1}^{n}||\nabla f_i(\mathbf{X}_0)||^2}, \nonumber \\
	c_4 &= L + (2L\delta_{\mathbf{C},\mathbf{R}}+L)/\underline{\Psi}, \nonumber \\
	c_5 &= \left[1+\sigma_{\mathbf{C}}+2L + (1-\sigma_{\mathbf{R}})(2L\delta_{\mathbf{C},\mathbf{R}}+L)/ \underline{\Psi}\right] \sqrt{n}.\nonumber
\end{align}
\begin{lemma}\label{lem5}
	Under the conditions in Lemmas~\ref{lem1:perron}--\ref{lem4:ineq_sys1}, if the step size $\alpha$ satisfies
	\begin{align} \label{eq:a_bd_1}
		\alpha < \min \{\frac{(1-\sigma_{\mathbf{R}})(1-\sigma_{\mathbf{C}})}{(\sqrt{n}L+\sqrt{n}+3)L\delta_{\mathbf{R},\mathbf{C}}\delta_{\mathbf{C},\mathbf{R}}}, \frac{1}{\delta_{\mathbf{R},\mathbf{C}}}  \},
	\end{align}
	then we have the following upper bounds for $ \left \| \mathbf{X}_{k}-\mathbf{1}\bar{\mathbf{x}}_{k}^{\top} \right \|_{\mathbf{R}} $ and $ \left \| \mathbf{Y}_{k}-\pi_{\mathbf{C}}\hat{\mathbf{y}}_{k}^{\top} \right \|_{\mathbf{C}} $:
	\begin{align}\label{eq:bound_s1}
		\begin{aligned}
			&\begin{bmatrix}
				\left \| \mathbf{X}_{k} - \mathbf{1} \bar{\mathbf{x}}_{k}^{\top} \right \|_{\mathbf{R}} \\
				\left \| \mathbf{Y}_{k} - \pi_{\mathbf{C}} \hat{\mathbf{y}}_{k}^{\top} \right \|_{\mathbf{C}}
			\end{bmatrix} \\
			\preceq
			&\begin{bmatrix}
				{c_0} & {c_1} & {c_2} \\
				{c_3} & {c_4} & {c_5}
			\end{bmatrix}
			\begin{bmatrix}
				\lambda^{k} \\
				\alpha \sum_{t=0}^{k-1} \lambda^{t} \left \| \nabla f(\bar{\mathbf{x}}_{k-1-t}) \right \| \\
				\sum_{t=0}^{k-1} \lambda^{t} e_{k-1-t}
			\end{bmatrix}.
		\end{aligned}
	\end{align}
\end{lemma}
\begin{proof}
	See APPENDIX B.
\end{proof}

\begin{lemma}\label{lem6}
	Under Assumptions~\ref{ass1:primal_cvx}--\ref{ass4:trigger}, we have the following upper bounds for $\sum_{t=0}^{k-1}\left \|\nabla f (\bar{\mathbf{x}}_t)\right \| \left \| \mathbf{X}_{t} - \mathbf{1} \bar{\mathbf{x}}_{t}^{\top} \right \|_{\mathbf{R}}$ and $\sum_{t=0}^{k-1}\left \|\nabla f (\bar{\mathbf{x}}_t)\right \| \left \| \mathbf{Y}_{t} - \pi_{\mathbf{C}} \hat{\mathbf{y}}_t^{\top} \right \|_{\mathbf{C}}$:
	\begin{align}\label{eq:bound_s3_matrix}
		\begin{aligned}
			&\begin{bmatrix}
				\sum_{t=0}^{k-1}\left\|\nabla f (\bar{\mathbf{x}}_t)\right\| \left\| \mathbf{X}_{t} - \mathbf{1} \bar{\mathbf{x}}_{t}^{\top} \right\|_{\mathbf{R}} \\
				\sum_{t=0}^{k-1}\left\|\nabla f (\bar{\mathbf{x}}_t)\right\| \left\| \mathbf{Y}_{t} - \pi_{\mathbf{C}} \hat{\mathbf{y}}_t^{\top} \right\|_{\mathbf{C}}
			\end{bmatrix} \\
			\preceq
			&\begin{bmatrix}
				\dfrac{c_0}{1-\lambda} & c_0 & \dfrac{c_1}{1-\lambda} + \dfrac{c_2}{2} & \dfrac{c_2}{2} \\[2mm]
				\dfrac{c_3}{1-\lambda} & c_3 & \dfrac{c_4}{1-\lambda} + \dfrac{c_5}{2} & \dfrac{c_5}{2}
			\end{bmatrix}
			\begin{bmatrix}
				1 \\
				\sum_{t=0}^{k-1}\lambda^{t} \left\| \nabla f (\bar{\mathbf{x}}_t) \right\|^2 \\
				\alpha \sum_{t=0}^{k-1}\left\| \nabla f(\bar{\mathbf{x}}_t) \right\|^2 \\
				\dfrac{e_0 S_e}{\alpha (1-\lambda)^2}
			\end{bmatrix}.
		\end{aligned}
	\end{align}
	
	\begin{proof}
		See APPENDIX C.
	\end{proof}
	
\end{lemma}

\begin{lemma}\label{lem7}
	Under Assumptions~\ref{ass1:primal_cvx}--\ref{ass4:trigger}, we have the following upper bounds for $\sum_{t=0}^{k-1}\left \| \mathbf{X}_{t} - \mathbf{1} \bar{\mathbf{x}}_{t}^{\top} \right \|_{\mathbf{R}} ^2$ and $\sum_{t=0}^{k-1}\left \| \mathbf{Y}_{t} - \pi_{\mathbf{C}} \hat{\mathbf{y}}_t^{\top} \right \|_{\mathbf{C}}^2 $:
	\begin{align}\label{eq:bound_s4_matrix}
		\begin{aligned}
			&\begin{bmatrix}
				\sum_{t=0}^{k-1}\left\| \mathbf{X}_{t} - \mathbf{1} \bar{\mathbf{x}}_{t}^{\top} \right\|_{\mathbf{R}}^2 \\
				\sum_{t=0}^{k-1}\left\| \mathbf{Y}_{t} - \pi_{\mathbf{C}} \hat{\mathbf{y}}_t^{\top} \right\|_{\mathbf{C}}^2
			\end{bmatrix}\\
			\preceq
			&\frac{2}{(1-\lambda)^2}
			\begin{bmatrix}
				c_0^2 & c_1^2 \alpha^2 & c_2^2 \\
				c_3^2 & c_4^2 \alpha^2 & c_5^2
			\end{bmatrix}
			\begin{bmatrix}
				1 \\
				\sum_{t=0}^{k-2} \left\| \nabla f(\bar{\mathbf{x}}_t) \right\|^2 \\
				\sum_{t=0}^{k-2} e_t^2
			\end{bmatrix}
		\end{aligned}
		.	\end{align}
\end{lemma}
\begin{proof}
	See APPENDIX D.
\end{proof}

Building on the preceding assumptions and lemmas, now we establish the main convergence theorem for the dual problem, without requiring convexity. To facilitate the presentation, positive constants $b_i,~i=1,\ldots,4$ related to network structure and function smoothness are defined as follows:
\begin{align*}
	b_1 &= L \sqrt{n} \pi_{\mathbf{C}}^{\top} \pi_{\mathbf{R}} \delta_{2,R}, &
	b_2 &= \delta_{2,C}, \\
	b_3 &= 3L \delta_{2,C}^2, &
	b_4 &= 3L^3 n (\pi_{\mathbf{C}}^{\top} \pi_{\mathbf{R}})^2 \delta_{2,R}^2
\end{align*}

\begin{theorem}[Convergence of ET-PPG without Convexity]\label{thm1:ET-PPG_nonconvex}
	Under Assumptions~\ref{ass1:primal_cvx}--\ref{ass4:trigger} and the conditions in Lemmas~\ref{lem1:perron}--\ref{lem7}, when the step size $\alpha$ satisfies~\eqref{eq:a_bd_1} and
	\begin{align}
		\begin{aligned}\label{eq:a_bd_3}
			\alpha <& \left[ \frac{3L (\pi_{\mathbf{C}}^{\top} \pi_{\mathbf{R}})^2 }{2} + \frac{ c_1 b_1 + c_4b_2+c_0 b_1 + c_3b_2}{1-\lambda} \right. \\
			& \left. + \frac{c_1^2b_3 + c_4^2b_4}{(1-\lambda)^2} + \frac{(c_2 b_1+c_5b_2)}{2}\right]^{-1}\pi_{\mathbf{C}}^{\top} \pi_{\mathbf{R}},		
		\end{aligned}	
	\end{align}
	the following convergence results hold:
	\begin{enumerate}
		\item $\left \|\nabla f (\bar{\mathbf{x}}_k)\right \|^2$ converges to $0$ at a rate of $\mathcal{O}(1/k) $, i.e.,
		\begin{align}\label{eq:bound_s5}
			\begin{aligned}
				&\frac{1}{k} \sum_{t=0}^{k-1}\left \|\nabla f (\bar{\mathbf{x}}_t)\right \|^2 \\
				\le & \frac{f(\bar{\mathbf{x}}_{0})-f^*}{\gamma k} + \frac{(c_2 b_1+c_5b_2)e_0S_e}{2(1-\lambda)^2\gamma k} \\
				& + \alpha  \frac{c_0 b_1 + c_3b_2}{(1-\lambda)\gamma k}\left(1+\sum_{t=0}^{k_0} \left \| \nabla f (\bar{\mathbf{x}}_t) \right \|^2\right) \\
				& + \alpha^2 \frac{c_0^2b_3 + c_3^2b_4+ (c_2^2b_3+c_5^2b_4) e_0 S_e}{(1-\lambda)^2\gamma k}.
			\end{aligned}		
		\end{align}
		\item $\left \| \mathbf{X}_{k} - \mathbf{1} \bar{\mathbf{x}}_{k}^{\top} \right \|_{\mathbf{R}}^2$ converges to $0$ at a rate of $\mathcal{O}(1/k) $, i.e.,
		\begin{align}\label{eq:bound_s6}
			\begin{aligned}
				&\frac{1}{k}\sum_{t=0}^{k-1}\left \| \mathbf{X}_{t} - \mathbf{1} \bar{\mathbf{x}}_{t}^{\top} \right \|_{\mathbf{R}} ^2\\
				\le &\frac{2}{(1-\lambda)^2k}\left[c_0^2 + c_1^2\alpha^2 \sum_{t=0}^{k-2}\left \| \nabla f(\bar{\mathbf{x}}_t) \right \|^2 + c_2^2 \sum_{t=0}^{k-2} e_t^2 \right].
			\end{aligned}
		\end{align}
	\end{enumerate}
\end{theorem}
\begin{proof}
	See APPENDIX E.	
\end{proof}

Theorem~\ref{thm1:ET-PPG_nonconvex} shows that in ET-PPG algorithm, $f(\bar{\mathbf{x}}_k)$ converges to a stationary point at a rate of $\mathcal{O}(1/k)$ without convexity. Then, we specialize the convergence of dual objective $f$ under the Polyak-Łojasiewicz (P-Ł) condition. The main theorem is provided after the P-Ł assumption and two supporting lemmas.

\begin{assumption}\label{ass5:PL}
	The global dual objective function $f$ adheres to the P-Ł condition with a positive constant $\beta$, i.e., $\forall \mathbf{x} \in \mathbb{R}^m$,
	\begin{align}
		\left \| \nabla f(\mathbf{x}) \right\|^2 \ge 2 \beta (f(\mathbf{x})-f^*).
	\end{align}
\end{assumption}

Similar to Lemma 4, a set of useful linear inequalities is first introduced, characterized by the following positive constants $d_i, i=1,\ldots,13$:
	\begin{align*}
	d_1  &= 4nL^2,                     &d_2  &= 2 \frac{3-\sigma_{\mathbf{R}}^2}{1-\sigma_{\mathbf{R}}^2} \delta_{\mathbf{R},\mathbf{C}}^2, \\
	d_3  &= 8nL, & d_4  &= 32\frac{1+\sigma_{\mathbf{C}}^2}{1-\sigma_{\mathbf{C}}^2}  L^2 \delta_{\mathbf{C},\mathbf{R}}^2, \\
	d_5  &= 16\frac{1+\sigma_{\mathbf{C}}^2}{1-\sigma_{\mathbf{C}}^2}nL^4 \delta_{\mathbf{C},\mathbf{R}}^2, 
	& d_6  &= 8 \frac{1+\sigma_{\mathbf{C}}^2}{1-\sigma_{\mathbf{C}}^2}L^2, \\
	d_7  &= 32 \frac{1+\sigma_{\mathbf{C}}^2}{1-\sigma_{\mathbf{C}}^2} n L^3, 
	& d_8  &= \frac{1}{2}L^2 n (\pi_{\mathbf{C}}^{\top} \pi_{\mathbf{R}})^2 \delta_{2,R}^2, \\
	d_9  &= \frac{3}{2}L^3 n (\pi_{\mathbf{C}}^{\top} \pi_{\mathbf{R}})^2 \delta_{2,R}^2, 
	& d_{10} &= \frac{1}{2}\delta_{2,C}^2, \\
	d_{11} &= \frac{3L}{2}\delta_{2,C}^2,            & d_{12} &= 2\beta\pi_{\mathbf{C}}^{\top} \pi_{\mathbf{R}}, \\
	d_{13} &= 2L+3(\pi_{\mathbf{C}}^{\top} \pi_{\mathbf{R}})^2 L^2.
\end{align*}

\begin{lemma}\label{lem8:ineq_sys2}
	Under Assumptions~\ref{ass1:primal_cvx}--\ref{ass5:PL}, the following linear system of inequalities w.r.t. $\left \|  \mathbf{X}_{k+1} - \mathbf{1} \bar{\mathbf{x}}_{k+1}^{\top} \right \|_{\mathbf{R}}^2$, $\left \| \mathbf{Y}_{k+1}-\pi_{\mathbf{C}}\hat{\mathbf{y}}_{k+1}^{\top} \right \|_{\mathbf{C}}^2$, and $f(\bar{\mathbf{x}}_{k+1}) - f^*$ holds:
	\begin{align}\label{eq:inqeasys_PL}
		\begin{bmatrix}
			\left \|  \mathbf{X}_{k+1} - \mathbf{1} \bar{\mathbf{x}}_{k+1}^{\top} \right \|_{\mathbf{R}}^2 \\ \left \| \mathbf{Y}_{k+1}-\pi_{\mathbf{C}}\hat{\mathbf{y}}_{k+1}^{\top} \right \|_{\mathbf{C}}^2 \\ f(\bar{\mathbf{x}}_{k+1}) - f^*
		\end{bmatrix}
		\preceq \mathbf{P}	\begin{bmatrix}
			\left \|  \mathbf{X}_{k} - \mathbf{1} \bar{\mathbf{x}}_{k}^{\top} \right \|_{\mathbf{R}}^2 \\ \left \| \mathbf{Y}_{k}-\pi_{\mathbf{C}}\hat{\mathbf{y}}_{k}^{\top} \right \|_{\mathbf{C}}^2 \\ f(\bar{\mathbf{x}}_{k}) - f^*
		\end{bmatrix}  + \mathbf{Q}e_k^2,
	\end{align}
	where matrix $\mathbf{P}$ and vector $\mathbf{Q}$ are given by
	\begin{align*}
		\mathbf{P}=&\begin{bmatrix}
			\frac{3\sigma_{\mathbf{R}}^2-\sigma_{\mathbf{R}}^4}{1+\sigma_{\mathbf{R}}^2} +d_1\alpha^2 & d_2 \alpha  ^2 & d_3 \alpha^2 \\
			d_4 + d_5\alpha^2 & \frac{1+\sigma_{\mathbf{C}}^2}{2} +d_6\alpha^2 & d_7\alpha^2\\
			d_8 + d_9\alpha^2 & d_{10 }+ d_{11}\alpha^2 & 1-d_{12}\alpha+d_{13}\alpha^2
		\end{bmatrix}\\
		\mathbf{Q} =& \begin{bmatrix}
			\frac{3+\sigma_{\mathbf{R}}^2}{1+\sigma_{\mathbf{R}}^2} (1+\sigma_{\mathbf{R}})^2n\\ 
			2 \frac{1+\sigma_{\mathbf{C}}^2}{1-\sigma_{\mathbf{C}}^2} \left[ (1+ \sigma_{\mathbf{C}})^2 + 4 (1+\sigma_{\mathbf{R}})^2\delta_{\mathbf{C},\mathbf{R}}^2 L^2  \right]n \\ 
			0
		\end{bmatrix}.	
	\end{align*}
\end{lemma}
\begin{proof}
	See APPENDIX F.
\end{proof}

\begin{lemma}[Lemma 5 in~\cite{pu2021distributed}]\label{lem9}
	Given a non-negative, irreducible matrix $\mathbf{M}= [m_{ij}] \in \mathbb{R}^{3 \times 3}$ with $m_{11}, m_{22}, m_{33} < \lambda^*$ for some $\lambda^* >0$ A necessary and sufficient condition for $\rho(\mathbf{M})<\lambda^*$ is $\det{(\lambda^*\mathbf{I}-\mathbf{M})}>0$.
\end{lemma}

Building on the linear inequality system and the bounded spectral radius condition, we now establish linear convergence for the dual problem under the P-Ł condition. Denote $\boldsymbol{\nu} = [\nu_1, \nu_2, \nu_3]^{\top}$ as an eigenvector of $\mathbf{P}$ associated with the spectral radius $\lambda = \rho(\mathbf{P})$. The positive coefficients $h_i,~i=1,\ldots,8$ are defined as follows for subsequent use in the convergence analysis:
\begin{align*}
	h_1 = & \frac{(1-\sigma_{\mathbf{R}}^2)^2(1-\sigma_{\mathbf{C}}^2)}{16 (1+\sigma_{\mathbf{R}}^2)}d_{12},\\
	h_2 = & \frac{(1-\sigma_{\mathbf{R}}^2)^2(1-\sigma_{\mathbf{C}}^2)}{16(1+\sigma_{\mathbf{R}}^2)}d_{13} + d_3d_4d_{10} + \frac{1-\sigma_{\mathbf{C}}^2}{2}d_3d_8 \\
	&+ \frac{(1-\sigma_{\mathbf{R}}^2)^2}{1+\sigma_{\mathbf{R}}^2}d_7d_{10},\\
	h_3 = & d_2d_4d_{12},\\
	h_4 = & d_2d_7d_8 + d_3d_5d_{10} + d_3d_4d_{11} + \frac{1-\sigma_{\mathbf{C}}^2}{2}d_3d_9 \\
	&+ \frac{(1-\sigma_{\mathbf{R}}^2)^2}{1+\sigma_{\mathbf{R}}^2}d_7d_{11} + d_2d_5d_{12},\\
	h_5 = & d_2d_5d_{13},\\
	h_6 = & d_2d_7d_9 + d_3d_5d_{11},\\
	h_7 = & \left[\frac{3+\sigma_{\mathbf{R}}^2}{1+\sigma_{\mathbf{R}}^2} (1+\sigma_{\mathbf{R}})^2+2 \frac{1+\sigma_{\mathbf{C}}^2}{1-\sigma_{\mathbf{C}}^2} \right. \\
	&\left.\left( (1+ \sigma_{\mathbf{C}})^2 + 4 (1+\sigma_{\mathbf{R}})^2\delta_{\mathbf{C},\mathbf{R}}^2 L^2  \right)\right]n,\\
	h_8 =& \sqrt{3}\frac{\max_{1\le i\le 3}\nu_i}{\min_{1\le i\le 3}\nu_i}.
\end{align*}

\begin{theorem}[Convergence of ET-PPG with P-Ł condition]\label{thm2:ET-PPG_PL}
	Under Assumptions~\ref{ass1:primal_cvx}--\ref{ass5:PL}, if the step size $\alpha$ satisfies
	\begin{align}\label{eq:PL_a_bd}
		\begin{aligned}
			\alpha < \min &\left \{ \frac{1}{d_{12}}, \frac{1-\sigma_{\mathbf{R}}^2}{4\sqrt{n}L\sqrt{1+\sigma_{\mathbf{R}}^2}},  \frac{1-\sigma_{\mathbf{C}}^2}{4\sqrt{2}L\sqrt{1+\sigma_{\mathbf{C}}^2}},\right. \\
			&\left. \frac{2\beta\pi_{\mathbf{C}}^{\top} \pi_{\mathbf{R}}}{2L + 3L^2(\pi_{\mathbf{C}}^{\top} \pi_{\mathbf{R}})^2},\right. \\
			&\left. \frac{-h_2 + \sqrt{h_2^2 + 4h_1(h_3+h_4+h_5+h_6)}}{2(h_3+h_4+h_5+h_6)} \right \},
		\end{aligned}
	\end{align}	then $\lambda <1$. 
	When the event-triggering error $\{e_k\}_{k=0}^{\infty}$ satisfies $e_k \le Es^{k}$ with $s\in (\sqrt{\lambda}, 1)$ and $E \ge \frac{\left\| \mathbf{z}_0 \right \| (s^2-\lambda)}{h_7}$ where $\mathbf{z}_0 = [	\left \| \mathbf{X}_{0}-\mathbf{1}\bar{\mathbf{x}}_{0}^{\top} \right \|_{\mathbf{R}}^2, \left \| \mathbf{Y}_{0}-\pi_{\mathbf{C}}\hat{\mathbf{y}}_{0}^{\top} \right \|_{\mathbf{C}}^2, f(\bar{\mathbf{x}}_0) - f^*]^{\top}$, then both the consensus error $\left \| \mathbf{X}_{k}-\mathbf{1}\bar{\mathbf{x}}_{k}^{\top} \right \|_{\mathbf{R}}^2$ and the optimality gap $f(\bar{\mathbf{x}}_k)- f^*$ converges to zero linearly, i.e.,
	\begin{align}\label{eq:PL_linear_converge}
		\left \| \mathbf{X}_{k}-\mathbf{1}\bar{\mathbf{x}}_{k}^{\top} \right \|_{\mathbf{R}}^2 + f(\bar{\mathbf{x}}_k) - f^* \le \sqrt{2} \frac{h_7h_8E^2 s^{2k}}{s^2 - \lambda}.
	\end{align}
\end{theorem}

\begin{proof}
	See APPENDIX G.
\end{proof}

The dual objective function $f_i$ in~\eqref{p:dual} is convex, as it comprises the sum of a convex conjugate function and an affine function. Consequently, by invoking Theorem~\ref{thm1:ET-PPG_nonconvex}, we conclude that the dual variable $\mathbf{X}_k$ converges sublinearly to the optimal solution. Furthermore, if each local dual objective $f_i$ exhibits strong convexity, then the global objective $f$ satisfies the P-Ł condition. Under this assumption, Theorem 2 guarantees that $\mathbf{X}_k$ converges linearly to the optimal solution.

\subsection{Convergence of Primal Problem}
In this subsection, we further establish the convergence properties of the primal variable $\mathbf{W}_k$. Denote $\mathbf{W}_k = [\mathbf{w}_{1,k}, \ldots, \mathbf{w}_{n,k}]^{\top}$, $\mathbf{W}^* = [\mathbf{w}_{1}^*, \ldots, \mathbf{w}_{n}^*]^{\top}$, $\nabla \mathbf{F}(\mathbf{W}_k) = [\nabla F_1(\mathbf{w}_{1,k}), \ldots, \nabla F_n(\mathbf{w}_{n,k})]^{\top} $ and $\overline{\nabla F}_k = \frac{1}{n} \sum_{i=1}^{n} \nabla F_i(\mathbf{w}_{i,k})$. The following convergence analysis is portrayed by the Karush-Kuhn-Tucker (KKT) condition for optimality.

\begin{theorem}[Convergence of ET-DGT]\label{thm3:primal conv}
	Under Assumptions~\ref{ass1:primal_cvx}--\ref{ass4:trigger}, if the constant step size $\alpha$ is sufficiently small and satisfies the conditions in Theorem~\ref{thm1:ET-PPG_nonconvex} or Theorem~\ref{thm2:ET-PPG_PL}, then the primal variable sequence $\{\mathbf{W}_k\}_{k=0}^{\infty}$ in Algorithm~\ref{alg:etdgt} converges to the optimal solution $\mathbf{W}^*$.
\end{theorem}
\begin{proof}
	Since $F_i$~\eqref{p:primal} is $\frac{1}{L}-$strongly convex, the Lagrangian~\eqref{p:lag} is also strongly convex, i.e.,
	\begin{align*}
		\begin{aligned}
			\mathcal{L}(\mathbf{W}^*,\bar{\mathbf{x}}_k) \ge & \mathcal{L}(\mathbf{W}_k,\bar{\mathbf{x}}_k) + \nabla_{\mathbf{w}} \mathcal{L}(\mathbf{W}_k,\bar{\mathbf{x}}_k)^{\top} (\mathbf{W}^* - \mathbf{W}_k) \\
			& + \frac{1}{2L} \left \| \mathbf{W}^* - \mathbf{W}_k \right \|^2.
		\end{aligned}
	\end{align*} 
	Assumption~\ref{ass1:primal_cvx} implies that strong duality holds and $F^* = -f^* = -\mathcal{L}(\mathbf{W}^*, \mathbf{x})$ for any $\mathbf{x}$. Therefore, we can obtain that,
	\begin{align*}
		f(\bar{\mathbf{x}}_k) - f^* = & \mathcal{L}(\mathbf{W}^*,\bar{\mathbf{x}}_k) - \inf_{\mathbf{W} \in \mathcal{W}} \mathcal{L}(\mathbf{W},\bar{\mathbf{x}}_k) \\
		= & \mathcal{L}(\mathbf{W}^*,\bar{\mathbf{x}}_k) -\mathcal{L}(\mathbf{W}_k,\bar{\mathbf{x}}_k) \\
		\ge & \frac{1}{2L} \left \| \mathbf{W}^* - \mathbf{W}_k \right \|^2,
	\end{align*}
	where the inequality follows from the first-order optimality condition. After rearranging terms, we obtain that
	\begin{align}\label{eq:w_converge}
		\left \| \mathbf{W}_k - \mathbf{W}^* \right \|^2 \le 2L(f(\bar{\mathbf{x}}_k) - f^*).
	\end{align}
	Since $f(\bar{\mathbf{x}}_k)$ converges to $f^*$ as stated in Theorem~\ref{thm1:ET-PPG_nonconvex} and Theorem~\ref{thm2:ET-PPG_PL}, we can obtain that the primal variable $\mathbf{W}_k$ converges to the optimal solution $\mathbf{W}^*$ in our ET-DGT algorithm.	
\end{proof}

To analyze the convergence rate of ET-DGT, we follow standard practice by considering the unconstrained and differentiable case, where each local constraint set $\mathcal{W}_i = \mathbb{R}^m$ and each local cost function $F_i$ is differentiable for all $i \in \mathcal{N}$.

\begin{theorem}[Convergence rate of ET-DGT]\label{thm4:convrate}
	Assume $\mathcal{W}_i = \mathbb{R}^m$,~$\forall i \in \mathcal{N}$ and $F_i$,~$\forall i \in \mathcal{N}$ is differentiable. If each $F_i$ is $\frac{1}{L}$-strongly convex as stated in Assumption~\ref{ass1:primal_cvx}, then under conditions in Theorem~\ref{thm1:ET-PPG_nonconvex}, the primal variables $\{ \mathbf{w}_{i,k}\}_{k=0}^{\infty}$ exhibits a convergence rate of $\mathcal{O}(1/k)$:
	\begin{align}
		\begin{aligned}
			&\frac{1}{k} \sum_{t=0}^{k-1} \Bigg[\sum_{i=1}^n \left\| \nabla F_i\left(\mathbf{w}_{i,k}\right) - \overline{\nabla F}_k\right\|^2 + \left\|\sum_{i=1}^{n}\mathbf{w}_{i,k} - \mathbf{d}\right\|^{2}\Bigg] \\
			\le & \frac{2}{k} \sum_{t=0}^{k-1}\left \|\nabla f (\bar{\mathbf{x}}_t)\right \|^2 +  \frac{2n^2(1+L^2)}{k}\sum_{t=0}^{k-1} \left \| \mathbf{X}_{t} - \mathbf{1} \bar{\mathbf{x}}_{t}^{\top} \right \|_{\mathbf{R}} ^2.
		\end{aligned}
	\end{align}
	If each $F_i$ is also $\mu$-Lipschitz smooth, then under conditions in Theorem~\ref{thm2:ET-PPG_PL}, $\{ \mathbf{w}_{i,k}\}_{k=0}^{\infty}$ exhibits a convergence rate of $\mathcal{O}(\lambda^k)$:
	\begin{align}
		\left \| \mathbf{W}_k - \mathbf{W}^* \right \|^2 \le 2\sqrt{2}L \frac{h_7h_8E^2 s^{2k}}{s^2 - \lambda}.
	\end{align}
	
\end{theorem}
\begin{proof}
	Since $\mathcal{W}_i = \mathbb{R}^m$,~$\forall i \in \mathcal{N}$, \eqref{eq:et_ww} becomes an unconstrained problem, where the first-order optimality condition requires that $\mathbf{x}_{i,k} = -\tilde{\mathbf{w}}_{i,k} = -\nabla F_i (\mathbf{w}_{i,k}) $. Thus, 
	\begin{align*}
		&\sum_{i=1}^n \left\| \nabla F_i\left(\mathbf{w}_{i,k}\right) - \overline{\nabla F}_k\right\|^2 \\
		&= \sum_{i=1}^n \left\| \mathbf{x}_{i,k} - \frac{1}{n}\sum_{j=1}^n \mathbf{x}_{j,k}\right\|^2 = \left\|\left(\mathbf{I}-\frac{1}{n}\mathbf{1}\mathbf{1}^{\top}\right)\mathbf{X}_k\right\|_F^2\\
		&\leq 2n \left\| \left(\mathbf{I}-\mathbf{1}\pi_{\mathbf{R}}^{\top}\right)\mathbf{X}_k\right\|^2 + 2n \left\|\left(\mathbf{1}\pi_{\mathbf{R}}^{\top}-\frac{1}{n}\mathbf{1}\mathbf{1}^{\top}\right)\mathbf{X}_k\right\|^2\\
		&= 2n \left\| \mathbf{X}_k - \mathbf{1}\bar{\mathbf{x}}_k^{\top}\right\|^2 + 2\left\|\left(\frac{1}{n}\mathbf{1}\mathbf{1}^{\top}-\mathbf{1}\pi_{\mathbf{R}}^{\top}\right)(\mathbf{X}_k- \mathbf{1}\bar{\mathbf{x}}_k^{\top})\right\|^2\\
		&\leq 2n^2\left\| \mathbf{X}_k- \mathbf{1}\bar{\mathbf{x}}_k^{\top}\right\|^2,
	\end{align*}
	where the first inequality follows from Cauchy-Schwarz (C-S) inequality and the last inequality comes from $\left\|\frac{1}{n}\mathbf{1}\mathbf{1}^{\top}-\mathbf{1}\pi_{\mathbf{R}}^{\top}\right\|^2~ \le n-1$.
	Then based on first-order optimality condition and~\eqref{eq:grad_fi}, it holds that $\nabla f_i(\mathbf{x}_{i,k}) = -\mathbf{w}_{i,k} + \mathbf{d}_i$, thus
	\begin{align*}
		&\left\|\sum_{i=1}^{n}\mathbf{w}_{i,k} - \mathbf{d}\right\|^{2}  = \left\|\sum_{i=1}^{n}\nabla f_{i}\left(\mathbf{x}_{i,k}\right)\right\|^{2} \\
		\leq &2\left\|\nabla f(\bar{\mathbf{x}}_{k})\right\|^{2} + 2n\sum_{i=1}^{n}\left\|\nabla f_{i}(\mathbf{x}_{i,k}) - \nabla f_{i}(\bar{\mathbf{x}}_{k})\right\|^{2} \\
		\leq &2\left\|\nabla f(\bar{\mathbf{x}}_{k})\right\|^{2} + 2n^2L^2\left\|\mathbf{X}_{k} - \mathbf{1}\bar{\mathbf{x}}_{k}^{\mathsf{T}}\right\|^{2},
	\end{align*}
	where the first inequality follows from C-S inequality and the last inequality follows from $L$-Lipschitz smoothness of $f_i$. Combining the previous equations, we obtain that
	\begin{align}\label{eq:F_sublinear}
		\begin{aligned}
			&\sum_{i=1}^n \left\| \nabla F_i\left(\mathbf{w}_{i,k}\right) - \overline{\nabla F}_k\right\|^2 + \left\|\sum_{i=1}^{n}\mathbf{w}_{i,k} - \mathbf{d}\right\|^{2} \\
			\le & 2\left\|\nabla f(\bar{\mathbf{x}}_{k})\right\|^{2} + 2n^2(1+L^2)\left\|\mathbf{X}_{k} - \mathbf{1}\bar{\mathbf{x}}_{k}^{\mathsf{T}}\right\|^{2}.
		\end{aligned}
	\end{align}
	The sublinear convergence of primal problem is established by integrating the results from~\eqref{eq:bound_s5},~\eqref{eq:bound_s6} and~\eqref{eq:F_sublinear}.
	If each function $F_i$ possesses a $\mu$-Lipschitz continuous gradient, by Fenchel duality, the corresponding function $f_i$ is $\frac{1}{\mu}$-strongly convex. Consequently, the aggregate function $f$ satisfies the P-Ł condition with $\beta =\frac{n}{\mu}$. The linear convergence of primal problem then follows by combining equations~\eqref{eq:PL_linear_converge} and~\eqref{eq:w_converge}.
\end{proof}

The convergence analysis addresses a key challenge: the dual of a strongly convex primal problem is not necessarily strongly convex. To overcome this, the analysis proceeds in two stages. First, we establish general convergence guarantees for ET-PPG, proving sublinear convergence for non-convex objectives (Theorem~\ref{thm1:ET-PPG_nonconvex}) and linear convergence under P-Ł condition (Theorem~\ref{thm2:ET-PPG_PL}). These results are then applied to the specific ET-DGT algorithm. We show that primal convergence follows from dual convergence (Theorem~\ref{thm3:primal conv}). Finally, we derive explicit rates for ET-DGT: if each $F_i$ is $\frac{1}{L}$-strongly convex, Theorem 1 yields a sublinear rate; if, in addition, each $F_i$ has a $\mu$-Lipschitz continuous gradient, then by Fenchel duality, the global dual objective $f$ satisfies P-Ł condition with $\beta = \frac{n}{\mu}$, enabling linear convergence via Theorem~\ref{thm2:ET-PPG_PL} (Theorem~\ref{thm4:convrate}). These results match the convergence rates of its periodic counterpart DDGT~\cite{zhang2020distributed}, while offering improved communication efficiency, thus enhancing the practical scalability of ET-DGT.

\section{Simulation}
\label{sec:simulation}
This section validates the proposed ET-DGT algorithm through numerical experiments, using the economic dispatch problem (EDP) in smart grids---a standard benchmark for distributed resource allocation. EDP aims to minimize total generation costs while ensuring a global supply-demand balance and respecting local generator capacity constraints~\cite{wang2021push}. For comprehensive evaluation, ET-DGT is compared against several state-of-the-art dual-based methods designed for unbalanced directed graphs.
\begin{itemize}
	\item \textbf{DDGT}~\cite{zhang2020distributed}: A periodic communication variant of the proposed algorithm, serving as a baseline to quantify communication savings from the event-triggered mechanism.
	\item \textbf{NN-SURPLUS}~\cite{xu2017distributed}: A periodic algorithm based on nonnegative surplus consensus, representing an alternative architecture for DRA over digraphs. 
	\item \textbf{ET-NN-SURPLUS}: An event-triggered adaptation of NN-SURPLUS, used to assess the robustness of different ET frameworks under similar communication constraints.
\end{itemize}

The experiments are divided into three cases. Case 1 evaluates convergence and communication efficiency on the IEEE 14-bus system~\cite{yang2013consensus} with quadratic costs. Case 2 extends the analysis to non-quadratic cost functions on the same system. Case 3 examines scalability on the larger IEEE 118-bus system~\cite{fu2006ac}.

\subsection{IEEE 14-Bus System with Quadratic Costs}
The first simulation is conducted on the IEEE 14-bus system, where buses $\{ 1,2,3,6,8 \}$ host generators and the remaining buses serve as loads. The communication network is defined as $\mathcal{G} = (\mathcal{V}, \mathcal{E})$, where $\mathcal{V}$ includes both generator and load buses. The edge set is given by $\mathcal{E} = \{ (i,i+1), (i, i+2) | 1 \le i \le 12 \} \cup \{(13,14), (13,1), (14,1), (1,7), (2,8), (3,2), $
$(3,9), (4,10), (5,2), (5,11), (6,12)\}$. See~\cite{xing2014distributed} for illustration of the graph topology. Each generator $i$ is assigned a quadratic cost function:
\begin{equation*}
	F_i(w_i) = a_iw_i^2 + b_i w_i
\end{equation*}
with parameters and local capacity constraints $\mathcal{W}_i$ adopted from~\cite{huo2024differentially} and detailed in Table~\ref{tab:generator_parameters}. Local power demands are given as: $d_1 = 0\text{MW}$, $d_2 = 9\text{MW}$, $d_3 = 56\text{MW}$, $d_4 = 55\text{MW}$, $d_5 = 27\text{MW}$, $d_6 = 27\text{MW}$, $d_7 = 0\text{MW}$, $d_8 = 0\text{MW}$, $d_9 = 8\text{MW}$, $d_{10} = 24\text{MW}$, $d_{11} = 53\text{MW}$, $d_{12} = 46\text{MW}$, $d_{13} = 16\text{MW}$, and $d_{14} = 40\text{MW}$. The total demand is $\sum_{i=1}^{14} d_i = 361\text{MW}$, which is invisible to any individual agents. The centralized optimal solution used as the ground truth for error evaluation is $\textbf{w}^* = [76.7398, 85.6530, 59.1311, 68.9863, 70.4898]^{\top}$.

For the algorithm parameters, a constant step-size $\alpha = 0.02$ is used for both DDGT and ET-DGT. The dynamically decaying triggering thresholds are set to $e_k = 0.35\times 0.91^k$ for ET-DGT and $e_k = 0.5\times 0.96^k$ for ET-NN-SURPLUS. 

\begin{table}[h!]
	\centering
	\caption{Generator Parameters}
	\label{tab:generator_parameters}
	\begin{tabular}{ccccc}
		\toprule
		Generator & Bus & $a_i$ (MW$^2$h) & $b_i$ ($ \$/ $MWh) &  $\mathcal{W}_i$ (MW) \\
		\midrule
		1 & 1 & 0.04 & 2.0 & [0,80] \\
		2 & 2 & 0.03 & 3.0 & [0,90] \\
		3 & 3 & 0.035 & 4.0 & [0,70] \\
		4 & 6 & 0.03 & 4.0 & [0,70] \\
		5 & 8 & 0.04 & 2.5 & [0,80] \\
		\bottomrule
	\end{tabular}
\end{table}

The results confirm the effectiveness of the proposed method. As shown in Fig.~\ref{fig:case1_gen_converge}, the power output of each generator under ET-DGT converges precisely to their respective optimal values.

\begin{figure}[h!]
	\centering
	\includegraphics[width=0.9\linewidth]{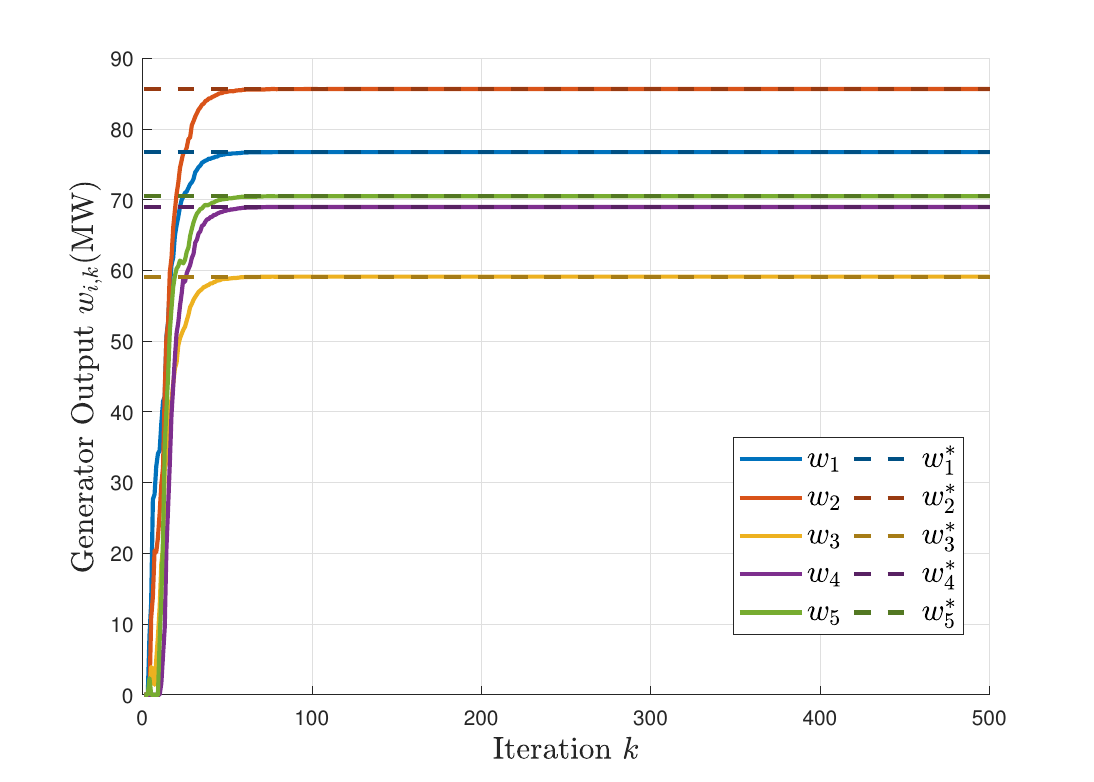}
	\caption{Convergence of individual generator outputs under the proposed ET-DGT algorithm in Case 1.}
	\label{fig:case1_gen_converge}
\end{figure}

A key performance metric is the satisfaction of global supply-demand balance. As shown in Fig.~\ref{fig:case1_gen_vs_demand}, total generation under ET-DGT converges smoothly to the total demand, matching the performance of periodic algorithms (DDGT, NN-SURPLUS). In contrast, ET-NN-SURPLUS exhibits a significant steady-state error, indicating constraint violation and underscoring the superior robustness of ET-DGT against event-triggered perturbations.

\begin{figure}[h!]
	\centering
	\includegraphics[width=0.9\linewidth]{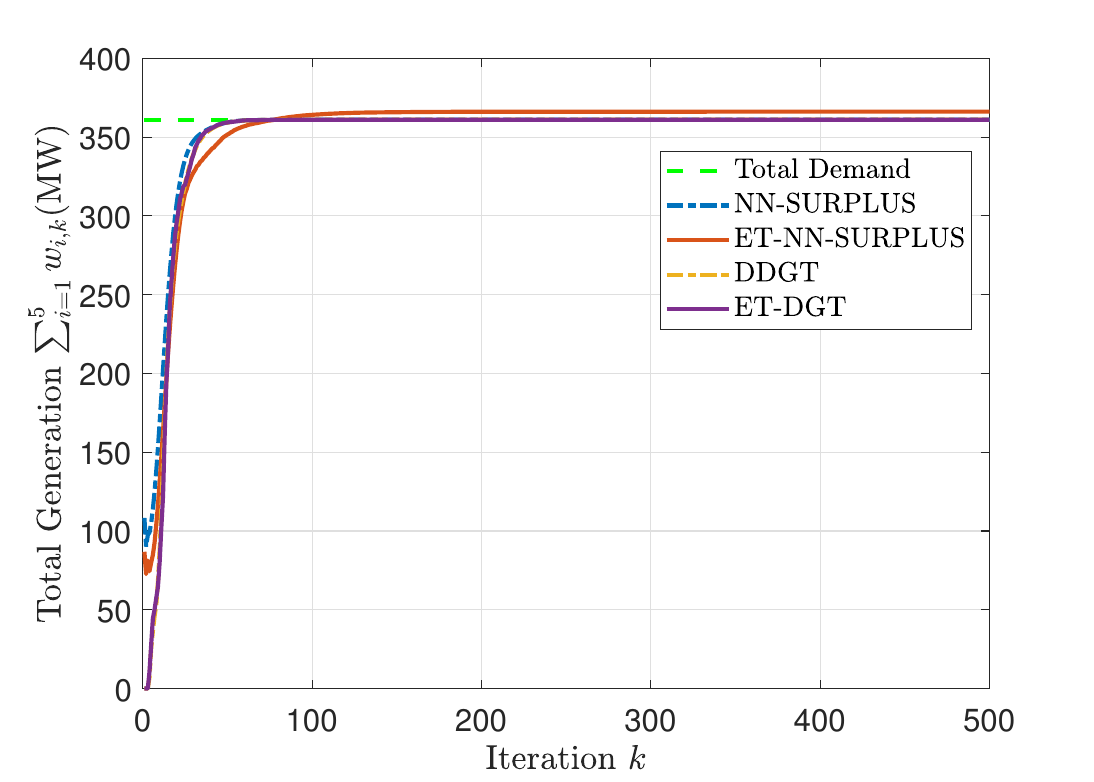}
	\caption{Total power generation v.s. total demand for all compared algorithms in Case 1.}
	\label{fig:case1_gen_vs_demand}
\end{figure}

Fig.~\ref{fig:case1_converge_error} illustrates the convergence error, supporting the core claims of this work. The trajectory of ET-DGT closely matches that of its periodic counterpart, DDGT, empirically validating that the event-triggered mechanism preserves both convergence rate and final accuracy. While NN-SURPLUS also converges effectively, its event-triggered variant converges more slowly and stabilizes at a significantly larger error, indicating reduced robustness to asynchronicity. 

\begin{figure}[h!]
	\centering
	\includegraphics[width=0.9\linewidth]{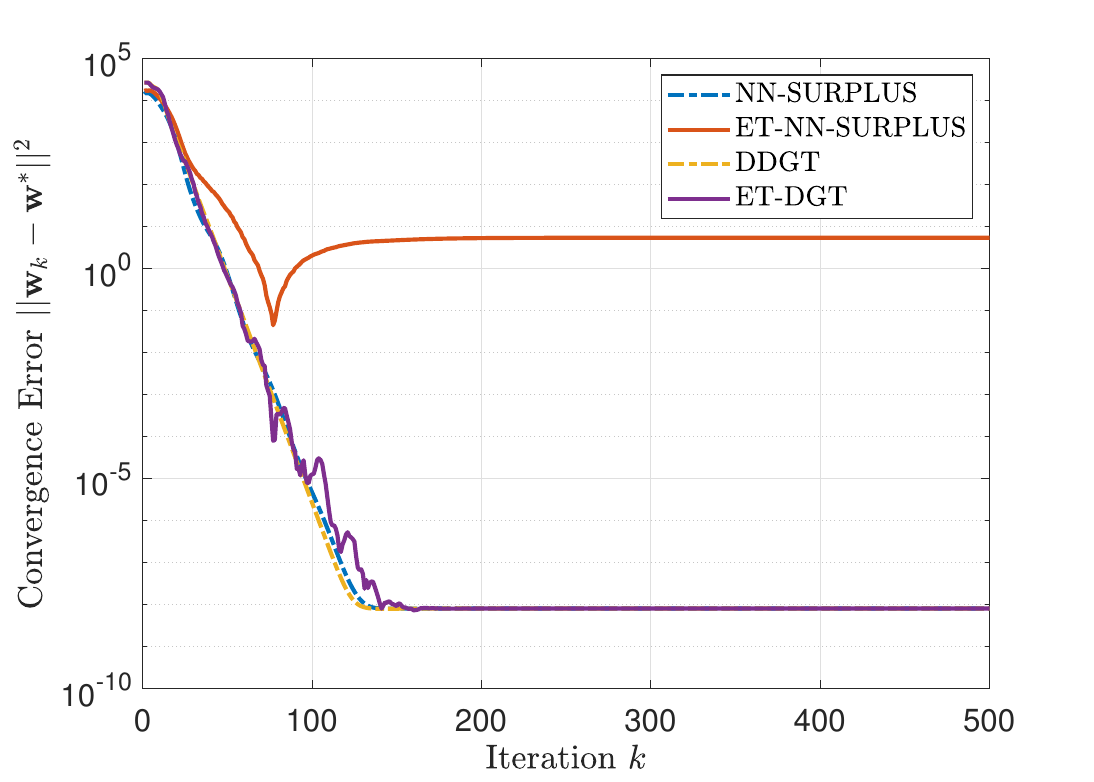}
	\caption{Convergence error comparison among algorithms in Case 1.}
	\label{fig:case1_converge_error}
\end{figure}

Crucially, ET-DGT achieves competitive performance with significantly reduced communication. As shown in Table~\ref{tab:communication counts}, ET-DGT uses only $58.1\%$ transmissions, demonstrating substantial communication savings without compromising accuracy. In contrast, ET-NN-SURPLUS consumes more resources ($64.3\%$) while yielding inferior performance.

\begin{table}[h!]
	\centering
	\caption{Comparison of total communication counts}
	\label{tab:communication counts}
	\begin{tabular}{cccc}
		\toprule
		Algorithm & Case 1 & Case 2 & Case 3\\
		\midrule
		DDGT & 7000 & 7000 & 118000\\
		\textbf{ET-DGT} & \textbf{4068} & \textbf{4702} & \textbf{54339}\\
		NN-SURPLUS & 7000 & 7000 & 118000\\
		ET-NN-SURPLUS & 4501 & 5383 & 65793\\
		\bottomrule
	\end{tabular}
\end{table}

\subsection{IEEE 14-Bus System with Exponential Costs}
To evaluate performance under more general convexity conditions, we introduce a non-quadratic cost function incorporating a natural exponential term~\cite{wang2021push}. This formulation yields a dual objective that may lack strong convexity, providing a practical test for our theoretical guarantees under the P-Ł condition. The cost function is defined as
\begin{equation*}
	F_i(w_i) = a_iw_i^2 + b_i w_i + d_i \exp{(\frac{w_i + e_i}{f_i})},
\end{equation*}
where $d_i = 1$, $e_i = 5$, and $f_i = 20$ for all $i \in {1,2,3,6,8}$. The centralized optimal solution is $\mathbf{w}^* = [74.4713, 76.9021, 67.5925, 70.0000, 72.0341]^{\top}$. Algorithm parameters are adjusted to $\alpha = 0.015$ and triggering thresholds $e_k = 0.03 \times 0.96^k$ for ET-DGT and $e_k = 0.5 \times 0.95^k$ for ET-NN-SURPLUS.

Fig.~\ref{fig:case2_converge_error} demonstrates that ET-DGT maintains convergence performance comparable to DDGT, even with more complex cost functions. Minor oscillations are observed due to aperiodic communication, yet the results empirically validate the theoretical guarantees without relying on strong convexity. The reference line $0.85^k$ indicates linear rate under non-strongly convex dual objectives. ET-NN-SURPLUS again converges more slowly and with higher error, reinforcing the robustness of push-pull gradient tracking under event-triggered communication.

\begin{figure}[h!]
	\centering
	\includegraphics[width=0.9\linewidth]{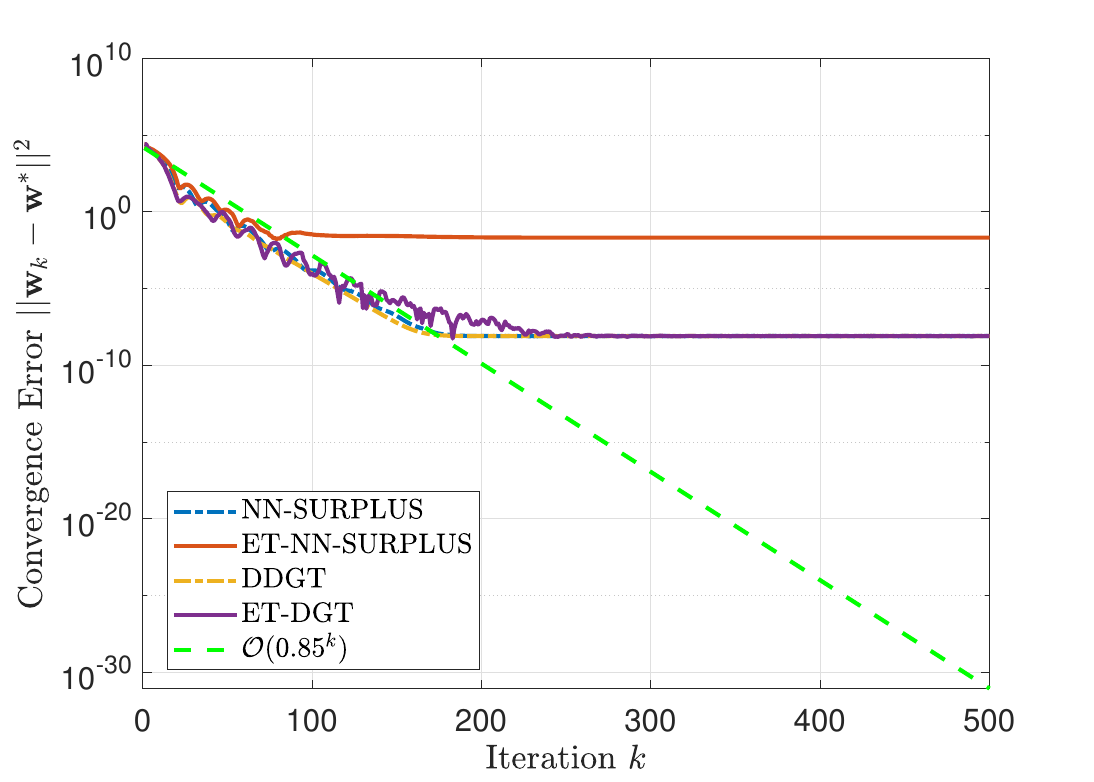}
	\caption{Convergence error comparison among algorithms in Case 2.}
	\label{fig:case2_converge_error}
\end{figure}

\subsection{IEEE 118-Bus System with Quadratic Costs}
To validate scalability, we apply the proposed algorithm to the EDP on the IEEE 118-bus system using MATPOWER~\cite{zimmerman2010matpower}, with simulations implemented in MATLAB 2023a and MATPOWER 8.1~\cite{matpower2025}. The communication topology is modeled as a randomly generated, strongly connected directed graph. For the larger network, algorithm parameters are set to $\alpha = 0.015$, with triggering thresholds $e_k = 0.5 \times 0.98^k$ for ET-DGT and $e_k = 0.985^k$ for ET-NN-SURPLUS.

Fig.~\ref{fig:case3_converge_error} shows that ET-DGT maintains performance comparable to DDGT even in large-scale settings, confirming its scalability and fast convergence. Notably, ET-DGT achieves this using only $46.1\%$ of DDGT’s communication resources (Table~\ref{tab:communication counts}), highlighting its efficiency at scale. In contrast, ET-NN-SURPLUS converges with large error, further highlighting the superior robustness of ET-DGT, particularly in complex network settings.

\begin{figure}[h!]
	\centering
	\includegraphics[width=0.9\linewidth]{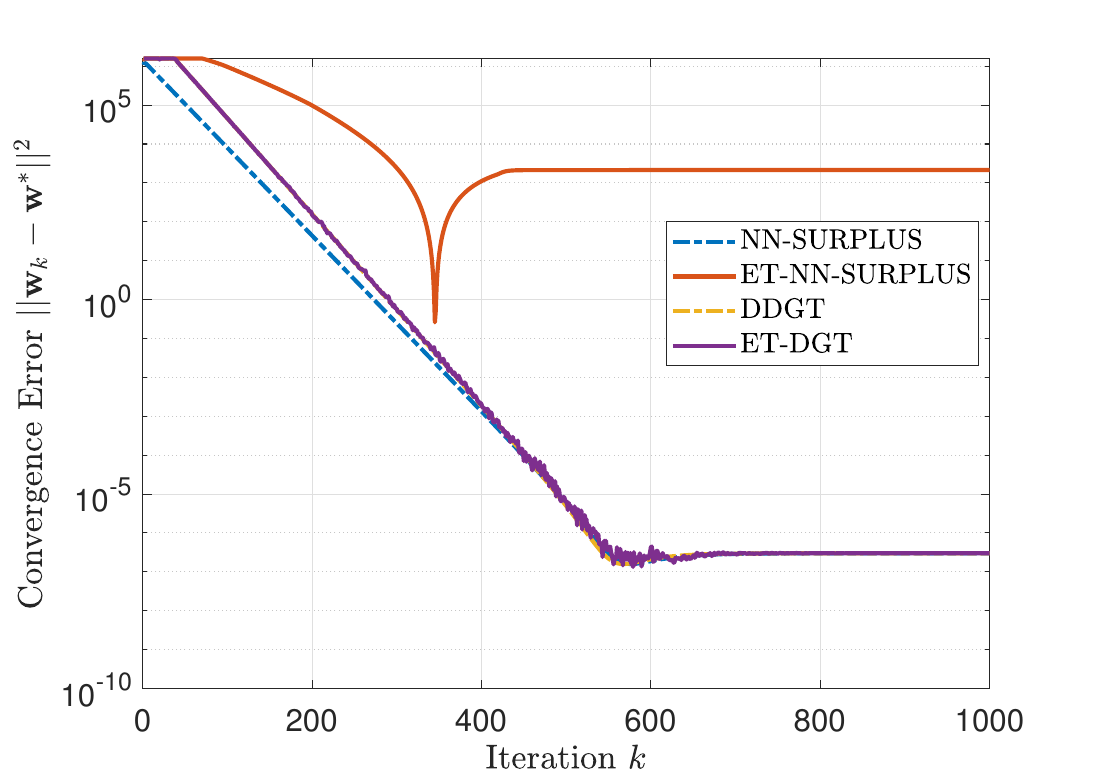}
	\caption{Convergence error comparison among algorithms in Case 3.}
	\label{fig:case3_converge_error}
\end{figure}

\section{Conclusion}
\label{sec:conclusion}
This paper proposed ET-DGT, an event-triggered dual gradient tracking algorithm for distributed resource allocation over unbalanced directed networks. By leveraging the duality between DRA and PPG, we developed a communication-efficient framework with rigorous convergence guarantees. Specifically, ET-PPG achieves linear convergence under the Polyak-Łojasiewicz condition and sublinear convergence for non-convex objectives. For ET-DGT, we established linear convergence with strongly convex and Lipschitz-smooth cost functions, and sublinear convergence without Lipschitz smoothness. Future research should address key challenges in DRA, including handling non-convex cost functions, ensuring communication privacy, fault tolerance, scalability in large networks, and adaptability to dynamic environments. In particular, designing locally adaptive event-triggering mechanisms remains an open problem for improving the communication–convergence trade-off.

\section{APPENDIX}
\subsection{Proof of Lemma 4} \label{apdx4_lemm4}
\textbf{Step 1: The first inequality w.r.t. $\left \| \mathbf{X}_{k+1}-\mathbf{1}\bar{\mathbf{x}}_{k+1}^{\top} \right \|_{\mathbf{R}}$:}\\
From~\eqref{eq:x}, we can derive that
\begin{align}\label{eq:x_con_err_1}
	& \left \| \mathbf{X}_{k+1}-\mathbf{1}\bar{\mathbf{x}}_{k+1}^{\top} \right \|_{\mathbf{R}} \nonumber \\
	= & \left\| (\mathbf{I}-\mathbf{1}\pi_{\mathbf{R}}^{\top}) \left [ \mathbf{R}\mathbf{X}_k- \alpha \mathbf{Y}_k  +  (\mathbf{R}-\mathbf{I})\boldsymbol{\zeta}_k \right ] \right \| \nonumber\\
	= & \left \| (\mathbf{R}-\mathbf{1}\pi_{\mathbf{R}}^{\top})(\mathbf{X}_{k}-\mathbf{1}\bar{\mathbf{x}}_{k}^{\top})  - (\mathbf{I}-\mathbf{1}\pi_{\mathbf{R}}^{\top}) \alpha \mathbf{Y}_k + (\mathbf{R}-\mathbf{I})\boldsymbol{\zeta}_k \right \| \nonumber\\
	\le & \sigma_{\mathbf{R}} \left \| \mathbf{X}_{k}-\mathbf{1}\bar{\mathbf{x}}_{k}^{\top} \right \|_{\mathbf{R}} + \alpha \left \| \mathbf{Y}_k \right \|_{\mathbf{R}} + (1+\sigma_{\mathbf{R}}) \sqrt{n} e_k,
\end{align}
where the inequality follows from Assumption~\ref{ass4:trigger} and the fact that $\left \| (\mathbf{I}-\mathbf{1}\pi_{\mathbf{R}}^{\top}) \right \|_{\mathbf{R}} = 1$ and $\left \| \mathbf{R}-\mathbf{I}\right \|_{\mathbf{R}} \le 1+\sigma_{\mathbf{R}}$ based on the construction of norm $\left \| \cdot \right \|_{\mathbf{R}}$ in Lemma~\ref{lem2:norm_bd}.\\
For the term regarding $\left \| \mathbf{Y}_k \right \|_{\mathbf{R}}$, we have
\begin{align}\label{eq:x_con_err_2}
	\left \| \mathbf{Y}_k \right \|_{\mathbf{R}} &= \left \| \mathbf{Y}_k -\pi_{\mathbf{C}}\hat{\mathbf{y}}_k^{\top} + \pi_{\mathbf{C}}\hat{\mathbf{y}}_k^{\top} \right \|_{\mathbf{R}} \nonumber \\
	& \le \delta_{\mathbf{R},\mathbf{C}}\left \| \mathbf{Y}_k -\pi_{\mathbf{C}}\hat{\mathbf{y}}_k^{\top} \right \|_{\mathbf{C}} +\left \| \hat{\mathbf{y}}_k^{\top} \right \|_{\mathbf{R}},
\end{align}
where the inequality follows from the fact that $ \left \| \pi_{\mathbf{C}} \right \|_{\mathbf{R}} \le 1$. 
Based on the $L$-Lipschitz smoothness of $f_i$, we have
\begin{align}\label{eq:x_con_err_3}
	\left \| \hat{\mathbf{y}}_k^{\top} \right \|_{\mathbf{R}} &= \left \| \mathbf{1}^{\top} \nabla\mathbf{f} (\mathbf{X}_{k})- \mathbf{1}^{\top} \nabla\mathbf{f} (\mathbf{1}\bar{\mathbf{x}}_k^{\top}) + \mathbf{1}^{\top} \nabla\mathbf{f} (\mathbf{1}\bar{\mathbf{x}}_k^{\top}) \right \|_{\mathbf{R}} \nonumber \\
	& \le \sqrt{n} L \left\| \mathbf{X}_k-\mathbf{1}\bar{\mathbf{x}}_k^{\top} \right \|_{\mathbf{R}} + \left \| \nabla f(\bar{\mathbf{x}}_k) \right \|_2.
\end{align}
Combining~\eqref{eq:x_con_err_1}--\eqref{eq:x_con_err_3}, we obtain the bound on $\left \| \mathbf{X}_{k+1}-\mathbf{1}\bar{\mathbf{x}}_{k+1}^{\top} \right \|_{\mathbf{R}}$ in~\eqref{eq:inequal_sys_1}.

\textbf{Step 2: The second inequality w.r.t. $\left \| \mathbf{Y}_{k+1}-\pi_{\mathbf{C}}\hat{\mathbf{y}}_{k+1}^{\top} \right \|_{\mathbf{C}}$:}\\
From~\eqref{eq:y}, we can obtain that
\begin{align}\label{eq:Y_k+1}
	& \mathbf{Y}_{k+1}-\pi_{\mathbf{C}}\hat{\mathbf{y}}_{k+1}^{\top}\\
	= & \left(\mathbf{I} -\pi_{\mathbf{C}}\mathbf{1}^{\top}\right) \left[ \mathbf{C}\mathbf{Y}_{k} +  \left(\mathbf{C} -\mathbf{I}\right) \boldsymbol{\xi}_k + \nabla\mathbf{f} (\mathbf{X}_{k+1})-\nabla\mathbf{f} (\mathbf{X}_{k}) \right]\nonumber
\end{align}
Therefore, we have 
\begin{align}
	& \left \| \mathbf{Y}_{k+1}-\pi_{\mathbf{C}}\hat{\mathbf{y}}_{k+1}^{\top} \right \|_{\mathbf{C}} \nonumber \\
	\le & \sigma_{\mathbf{C}} \left \| \mathbf{Y}_{k}-\pi_{\mathbf{C}}\hat{\mathbf{y}}_{k}^{\top} \right \|_{\mathbf{C}} + L  \left \| \mathbf{X}_{k+1} - \mathbf{X}_{k}  \right \|_{\mathbf{C}} \nonumber \\
	& + (1+\sigma_{\mathbf{C}}) \sqrt{n} e_k,\end{align}
where the inequality follows from the constructed $\left \| \cdot \right \|_{\mathbf{C}}$ and the $L$-Lipschitz smoothness of $f_i$.\\
For the term regarding $\left \| \mathbf{X}_{k+1}-\mathbf{X}_{k} \right \|_{\mathbf{C}}$, we first derive
\begin{align}\label{eq:x_k+1 - x_k}
	\mathbf{X}_{k+1}-\mathbf{X}_{k}
	= & \left(\mathbf{R} - \mathbf{I}\right) \mathbf{X}_k + \left(\mathbf{R} - \mathbf{I}\right) \boldsymbol{\zeta}_k \nonumber \\
	& -\alpha \left(\mathbf{Y}_{k} 
	-\pi_{\mathbf{C}}\hat{\mathbf{y}}_{k}^{\top}\right)-\alpha \pi_{\mathbf{C}}\hat{\mathbf{y}}_{k}^{\top}.
\end{align}
And then, 
\begin{align}\label{eq:xk+1-xk}
	& \left \| \mathbf{X}_{k+1}-\mathbf{X}_{k} \right \|_{\mathbf{C}} \nonumber \\
	\le & \left \| \mathbf{R} - \mathbf{I} \right \|_{\mathbf{C}} \left \| \mathbf{X}_k - \mathbf{1}\bar{\mathbf{x}}_k^{\top} \right \|_{\mathbf{C}} + \alpha \left  \| \mathbf{Y}_{k}-\pi_{\mathbf{C}}\hat{\mathbf{y}}_{k}^{\top} \right \|_{\mathbf{C}} \nonumber \\
	& + \alpha \left \| \pi_{\mathbf{C}} \right \|_{\mathbf{C}} \left \| \hat{\mathbf{y}}_{k}^{\top} \right \|_{\mathbf{C}} + \left \| \mathbf{R} - \mathbf{I}\right \|_{\mathbf{C}} \left \| \boldsymbol{\zeta}_k\right \|_{\mathbf{C}}\nonumber \\
	\le & \delta_{\mathbf{C},\mathbf{R}} \left( 2 + \alpha \sqrt{n} L \right) \left \| \mathbf{X}_k - \mathbf{1}\bar{\mathbf{x}}_k^{\top} \right \|_{\mathbf{R}} \nonumber \\
	& + \alpha \left \| \mathbf{Y}_{k}-\pi_{\mathbf{C}}\hat{\mathbf{y}}_{k}^{\top} \right \|_{\mathbf{C}} + \alpha  \left \| \nabla f(\bar{\mathbf{x}}_k) \right \|_2 + 2 \sqrt{n} e_k.
\end{align}
Combining~\eqref{eq:Y_k+1}--\eqref{eq:xk+1-xk}, we have the bound on $\left \| \mathbf{Y}_{k+1}-\pi_{\mathbf{C}}\hat{\mathbf{y}}_{k+1}^{\top} \right \|_{\mathbf{C}}$ in~\eqref{eq:inequal_sys_1}.
\subsection{Proof of Lemma 5}\label{apdx:lem5}
Note that for matrix $\mathbf{P}$ in Lemma 4, $\rho(\mathbf{P})<1$ for sufficiently small $\alpha$, since
\begin{align*}
	\lim_{\alpha \to 0} \mathbf{P} = \begin{bmatrix}
		\sigma_{\mathbf{R}} & 0 \\
		2L\delta_{\mathbf{C},\mathbf{R}} & \sigma_{\mathbf{C}}
	\end{bmatrix}
\end{align*}
is an lower triangular matrix. Its eigenvalues $\sigma_{\mathbf{R}}<1$ and $\sigma_{\mathbf{C}}<1$ are defined in Lemma 2.\\
The linear matrix inequalities~\eqref{eq:inequal_sys_1} implies that
\begin{align}\label{eq:lyap_1}
	\mathbf{z}_k \preceq  \mathbf{P}^{k} \mathbf{z}_{0} + \sum_{t=0}^{k-1} \mathbf{P}^{t} \mathbf{Q}\mathbf{u}_{k-1-t},
\end{align}
where
\begin{align*}
	\mathbf{z}_k \triangleq \begin{bmatrix}
		\left \| \mathbf{X}_{k}-\mathbf{1}\bar{\mathbf{x}}_{k}^{\top} \right \|_{\mathbf{R}} \\
		\left \| \mathbf{Y}_{k}-\pi_{\mathbf{C}}\hat{\mathbf{y}}_{k}^{\top} \right \|_{\mathbf{C}}
	\end{bmatrix}, \quad \text{and} \quad
	\mathbf{u}_k \triangleq \begin{bmatrix}
		\alpha \left \| \nabla f(\bar{\mathbf{x}}_k) \right \|_2 \\
		\sqrt{n}e_k
	\end{bmatrix}.
\end{align*}
Let $\lambda_1$ and $\lambda_2$ be the two eigenvalues of $\mathbf{P}$ such that $|\lambda_2| > |\lambda_1|$, then $\mathbf{P}$ can be diagonalized as
\begin{align} \label{eq:eigen_de_P}
	\mathbf{P} = \mathbf{T} \boldsymbol{\Lambda} \mathbf{T}^{-1}, \quad \boldsymbol{\Lambda} = \begin{bmatrix}
		\lambda_2 & 0 \\
		0 & \lambda_1
	\end{bmatrix}.
\end{align}
Let $\Psi = \sqrt{(p_{11}-p_{12})^2+4p_{12}p_{21}}$. Without loss of generality, we assume $\sigma_{\mathbf{R}} \neq \sigma_{\mathbf{C}}$, which guarantees that $\Psi$ is lower bounded by a positive value $\underline{\Psi}$ that is independent of $\alpha$:
\begin{align*}
	\Psi \ge \sqrt{\left (\sigma_{\mathbf{R}} - \sigma_{\mathbf{C}} \right )^2} = \underline{\Psi} > 0.
\end{align*}
Then the matrices $\Lambda$, $\mathbf{T}$ and $\mathbf{T}^{-1}$ can be calculated as:
\begin{align*}
	\boldsymbol{\Lambda} &=  \begin{bmatrix}
		\lambda_2 & 0 \\
		0 & \lambda_1
	\end{bmatrix} = \begin{bmatrix}
		\frac{p_{11}+p_{22}+\Psi}{2} & 0 \\
		0 & \frac{p_{11}+p_{22}-\Psi}{2}
	\end{bmatrix},\\
	\mathbf{T} &= \begin{bmatrix}
		\frac{p_{11}-p_{22}-\Psi}{2P_{21}} & \frac{p_{11}-p_{22}+\Psi}{2p_{21}}\\
		1 & 1
	\end{bmatrix},\\
	\mathbf{T}^{-1} &= \begin{bmatrix}
		-\frac{p_{21}}{\Psi} & \frac{p_{11}-p_{22}+\Psi}{2\Psi}\\
		\frac{p_{21}}{\Psi} & \frac{p_{22}-p_{11}+\Psi}{2\Psi}
	\end{bmatrix}.
\end{align*}
By enforcing $\lambda = \rho(\mathbf{P})=\lambda_2<1$, we obtain that
\begin{align*}
	\alpha < & \frac{(1-\sigma_{\mathbf{R}})(1-\sigma_{\mathbf{C}})}{(\sqrt{n}+1)L\delta_{\mathbf{R},\mathbf{C}}+(\sqrt{n}L+2)L\delta_{\mathbf{R},\mathbf{C}}\delta_{\mathbf{C},\mathbf{R}}} \nonumber \\
	< & \frac{(1-\sigma_{\mathbf{R}})(1-\sigma_{\mathbf{C}})}{(\sqrt{n}L+\sqrt{n}+3)L\delta_{\mathbf{R},\mathbf{C}}\delta_{\mathbf{C},\mathbf{R}}}.
\end{align*}
If $\alpha$ additionally satisfies $\alpha < \frac{1}{\delta_{\mathbf{R},\mathbf{C}}}$, we can obtain that:
\begin{align*}
	& \frac{\lambda_1^k+\lambda_2^k}{2} + \frac{(p_{11}-p_{22})(\lambda_2^k-\lambda_1^k)}{2\Psi}  \\
	\le & \frac{\lambda_1^k+\lambda_2^k}{2} + \frac{|p_{11}-p_{22}|(\lambda_2^k-\lambda_1^k)}{2\sqrt{(p_{11}-p_{12})^2+4p_{12}p_{21}}} \le  \lambda ^k,\\
	& \frac{\lambda_1^k+\lambda_2^k}{2} + \frac{(p_{11}-p_{22})(\lambda_1^k-\lambda_2^k)}{2\Psi}  \\
	\le & \frac{\lambda_1^k+\lambda_2^k}{2} + \frac{|p_{11}-p_{22}|(\lambda_2^k-\lambda_1^k)}{2\sqrt{(p_{11}-p_{12})^2+4p_{12}p_{21}}} 
	\le \lambda ^k,
\end{align*}
and
\begin{align*}
	\frac{p_{12}}{\Psi}(\lambda_2^k-\lambda_1^k) 
	= &\frac{\alpha \delta_{\mathbf{R},\mathbf{C}}}{\Psi}(\lambda_2^k-\lambda_1^k)
	\le \frac{1}{\underline{\Psi}} \lambda^k, \\
	\frac{p_{21}}{\Psi}(\lambda_2^k-\lambda_1^k) 
	\le &  \frac{ 2L \delta_{\mathbf{C},\mathbf{R}} + \frac{(1-\sigma_{\mathbf{R}})(1-\sigma_{\mathbf{C}})\sqrt{n}L}{(\sqrt{n}L+\sqrt{n}+3)\delta_{\mathbf{R},\mathbf{C}}} }{\underline{\Psi}}\lambda^k \nonumber \\
	\le & \frac{ 2L \delta_{\mathbf{C},\mathbf{R}} + L }{\underline{\Psi}}\lambda^k.	
\end{align*}
It then follows from~\eqref{eq:eigen_de_P} that
\begin{align}\label{eq:P^k}
	0 &\preceq \mathbf{P}^k \nonumber \\
	&= \mathbf{T} \Lambda^k \mathbf{T}^{-1} \nonumber \\
	&=\begin{bmatrix}
		\frac{\lambda_1^k+\lambda_2^k}{2} + \frac{(p_{11}-p_{22})(\lambda_2^k-\lambda_1^k)}{2\Psi} & \frac{p_{12}}{\Psi}(\lambda_2^k-\lambda_1^k) \\	
		\frac{p_{21}}{\Psi}(\lambda_2^k-\lambda_1^k) & \frac{\lambda_1^k+\lambda_2^k}{2} + \frac{(p_{11}-p_{22})(\lambda_1^k-\lambda_2^k)}{2\Psi}
	\end{bmatrix} \nonumber \\
	&\preceq \lambda^k \begin{bmatrix}
		1 & \frac{1}{\underline{\Psi}}\\
		\frac{(2L\delta_{\mathbf{C},\mathbf{R}}+L)}{\underline{\Psi}}  & 1
	\end{bmatrix}.
\end{align}
Combining \eqref{eq:inequal_sys_1}, \eqref{eq:lyap_1} and \eqref{eq:P^k} yields~\eqref{eq:bound_s1}.

\subsection{Proof of Lemma 6}\label{apdx:lem6}

For $\left \| \nabla f(\bar{\mathbf{x}}_t) \right \| \left \| \mathbf{X}_{t} - \mathbf{1} \bar{\mathbf{x}}_{t}^{\top} \right \|_{\mathbf{R}}$, it follows from~\eqref{eq:bound_s1} that
\begin{align}
	&\left \| \nabla f(\bar{\mathbf{x}}_t) \right \| \left \| \mathbf{X}_{t} - \mathbf{1} \bar{\mathbf{x}}_{t}^{\top} \right \|_{\mathbf{R}} \nonumber \\
	\le & c_0 \lambda^{t} \left \| \nabla f(\bar{\mathbf{x}}_t) \right \| + c_1 \alpha \left \| \nabla f(\bar{\mathbf{x}}_t) \right \| \sum_{s=0}^{t-1} \lambda^{s} \left \| \nabla f(\bar{\mathbf{x}}_{t-1-s}) \right \| \nonumber \\
	& + c_2 \left \| \nabla f(\bar{\mathbf{x}}_t) \right \| \sum_{s=0}^{t-1} \lambda^{s} e_{t-1-s}.
\end{align}
Thus, we obtain that
\begin{align} \label{eq:bound_s3_x_0}
	& \sum_{t=0}^{k-1}\left \|\nabla f (\bar{\mathbf{x}}_t)\right \| \left \| \mathbf{X}_{t} - \mathbf{1} \bar{\mathbf{x}}_{t}^{\top} \right \|_{\mathbf{R}} \nonumber \\
	\le & c_0 \sum_{t=0}^{k-1}\lambda^{t} \left \| \nabla f(\bar{\mathbf{x}}_t) \right \| \nonumber \\
	& + c_1 \alpha \sum_{t=0}^{k-1}\left \| \nabla f(\bar{\mathbf{x}}_t) \right \| \sum_{s=0}^{t-1} \lambda^{s} \left \| \nabla f(\bar{\mathbf{x}}_{t-1-s}) \right \| \nonumber \\
	& + c_2 \sum_{t=0}^{k-1}\left \| \nabla f(\bar{\mathbf{x}}_t) \right \| \sum_{s=0}^{t-1} \lambda^{s} e_{t-1-s}.
\end{align}
For the first term in~\eqref{eq:bound_s3_x_0}, the relation $\left \| \nabla f (\bar{\mathbf{x}}_t) \right \| \le 1 + \left \| \nabla f (\bar{\mathbf{x}}_t) \right \|^2$ yields that 
\begin{align} \label{eq:s3_x_1}
	c_0 \sum_{t=0}^{k-1}\lambda^{t} \left \| \nabla f(\bar{\mathbf{x}}_t) \right \| \le & c_0 \sum_{t=0}^{k-1}\lambda^{t} \left(1 + \left \| \nabla f (\bar{\mathbf{x}}_t) \right \|^2\right) \nonumber \\
	\le & \frac{c_0}{1-\lambda} + c_0 \sum_{t=0}^{k-1}\lambda^{t} \left \| \nabla f (\bar{\mathbf{x}}_t) \right \|^2.
\end{align}
For the second term in~\eqref{eq:bound_s3_x_0}, we define
\begin{align*}
	\boldsymbol{\vartheta}_t =& \left[ \lambda^{t-1}, \lambda^{t-2}, \ldots, \lambda, 1, 0, \ldots, 0 \right]^{\top} \in \mathbb{R}^{k}, \\
	\tilde{\boldsymbol{\vartheta}}_t =& \left[ \underbrace{0, \dots, 0}_{t}, 1, 0, \dots, 0 \right]^{\top} \in \mathbb{R}^{k}, \\
	\mathbf{g}_k =& \left[ \left \| \nabla f(\bar{\mathbf{x}}_0) \right \|, \ldots, \left \| \nabla f(\bar{\mathbf{x}}_{k-1}) \right \| \right]^{\top} \in \mathbb{R}^{k}, \\
	\Theta_k =& \sum_{t=0}^{k-1}\tilde{\boldsymbol{\vartheta}}_t \boldsymbol{\vartheta}_t^{\top} 
	= \begin{bmatrix}
		0  \\
		1 & 0  \\
		\lambda & 1 & 0\\
		\vdots & \vdots & \ddots & \ddots\\
		\lambda^{k-2} & \lambda^{k-3}& \cdots & 1  & 0 
	\end{bmatrix}.
\end{align*}
Thus, it follows that
\begin{align}\label{eq:s3_x_2}
	& c_1 \alpha \sum_{t=0}^{k-1}\left \| \nabla f(\bar{\mathbf{x}}_t) \right \| \sum_{s=0}^{t-1} \lambda^{s} \left \| \nabla f(\bar{\mathbf{x}}_{t-1-s}) \right \| \nonumber \\
	= & c_1 \alpha \mathbf{g}_k^{\top} \Theta_k \mathbf{g}_k \nonumber \\
	\le & c_1 \alpha \frac{\rho(\Theta_k + \Theta_k^{\top})}{2} \left\| \mathbf{g}_k \right \|^2 = \frac{c_1 \alpha}{1-\lambda} \sum_{t=0}^{k-1}\left \| \nabla f (\bar{\mathbf{x}}_t) \right \|^2.
\end{align}
While for the third term in~\eqref{eq:bound_s3_x_0}, let $S_t = \sum_{s=0}^{t-1} \lambda^{s} e_{t-1-s} = \sum_{m=0}^{t-1} \lambda^{t-1-m}e_m$, then we can derive the following upper bounds w.r.t. $S_t^2$:

\begin{align}
	S_t^2 =& \left(\sum_{m=0}^{t-1} \lambda^{t-1-m}e_m\right)^2 \le \frac{e_0}{1-\lambda}	\left(\sum_{m=0}^{t-1} \lambda^{t-1-m}e_m\right),
\end{align}
and thus
\begin{align}\label{eq:St2}
	\sum_{t=0}^{k-1}S_t^2 \le & \frac{e_0}{1-\lambda} \sum_{t=0}^{k-1}\sum_{m=0}^{t-1}\left(\lambda^{t-1-m}e_m\right) \nonumber \\
	= & \frac{e_0}{1-\lambda} \sum_{m=0}^{k-2} e_m \sum_{t=m+1}^{k-1}\lambda^{t-1-m} \le \frac{e_0S_e}{(1-\lambda)^2}.
\end{align}
Therefore, we can derive that
\begin{align}\label{eq:s3_x_3}
	& c_2 	\sum_{t=0}^{k-1} \left \| \nabla f(\bar{\mathbf{x}}_t) \right \| S_t \nonumber \\
	\le& \frac{\alpha c_2}{2}  	\sum_{t=0}^{k-1} \left \| \nabla f(\bar{\mathbf{x}}_t) \right \|^2 + \frac{c_2}{2\alpha } 	\sum_{t=0}^{k-1} S_t^2\nonumber \\
	\le&  \frac{\alpha c_2}{2} 	\sum_{t=0}^{k-1}\left \| \nabla f(\bar{\mathbf{x}}_t) \right \|^2 + \frac{c_2e_0S_e}{2\alpha (1-\lambda)^2},
\end{align}
where the first inequality follows from Young's inequality and the second inequality follows from~\eqref{eq:St2}. Combining~\eqref{eq:bound_s3_x_0}--\eqref{eq:s3_x_2} and~\eqref{eq:s3_x_3} gives the bound on $\sum_{t=0}^{k-1} \left \|\nabla f (\bar{\mathbf{x}}_t)\right \| \left \| \mathbf{X}_{t} - \mathbf{1} \bar{\mathbf{x}}_{t}^{\top} \right \|_{\mathbf{R}}$. Similarly, we can bound $ \sum_{t=0}^{k-1} \left \|\nabla f (\bar{\mathbf{x}}_t)\right \| \left \| \mathbf{Y}_{t} - \pi_{\mathbf{C}} \hat{\mathbf{y}}_t^{\top} \right \|_{\mathbf{C}} $ and obtain~\eqref{eq:bound_s3_matrix}.

\subsection{Proof of Lemma 7}\label{apdx:lem7}
We first define
\begin{align*}
	\mathbf{v}_{k} =& \left[c_0, c_1 \alpha \left \| \nabla f(\bar{\mathbf{x}}_0) \right \|, \ldots, c_1 \alpha \left \| \nabla f(\bar{\mathbf{x}}_{k-2}) \right \| \right]^{\top} \in \mathbb{R}^{k}, \\
	\boldsymbol{\epsilon}_{k} =& \left[0, c_2e_0, \ldots, c_2e_{k-2}\right]^{\top} \in \mathbb{R}^{k}, \\
	\boldsymbol{\phi}_t =& \left[\lambda^{t}, \lambda^{t-1}, \ldots, \lambda, 1, 0, \ldots, 0 \right]^{\top} \in \mathbb{R}^{k}, \\
	\Phi_k =& \sum_{t=0}^{k-1}\boldsymbol{\phi}_t\boldsymbol{\phi}_t^{\top} \in \mathbb{R}^{{k}\times {k}}.
\end{align*}

Combining~\eqref{eq:bound_s1}, we have that $\left \| \mathbf{X}_{t} - \mathbf{1} \bar{\mathbf{x}}_{t}^{\top} \right \|_{\mathbf{R}} \le \mathbf{v}_{k}^{\top} \boldsymbol{\phi}_t + \boldsymbol{\epsilon}_{k}^{\top} \boldsymbol{\phi}_t$ and thus,

\begin{align}
	\sum_{t=0}^{k-1} \left \| \mathbf{X}_{t} - \mathbf{1} \bar{\mathbf{x}}_{t}^{\top} \right \|_{\mathbf{R}} ^2 \le& 2 \mathbf{v}_{k}^{\top}\Phi_k \mathbf{v}_{k} + 2 \boldsymbol{\epsilon}_k^{\top}\Phi_k \boldsymbol{\epsilon}_k \nonumber \\
	\le& 2 \left\| \Phi_k \right\| \left(\left \| \mathbf{v}_{k} \right \|^2 + \left \| \boldsymbol{\epsilon}_{k} \right \|^2\right).
\end{align}
To bound $\left\| \Phi_k \right\|$, let $[\Phi_k]_{ij}$ be the $(i,j)$-th element. For any $0<i\le j \le k$, we have,
\begin{align}
	[\Phi_k]_{ij} =& \sum_{t=j-1}^{k-1} \lambda^{t-i+1} \lambda^{t-j+1} =  \frac{\lambda^{j-i}\left(1-\lambda^{2(k-j+1)}\right)}{1-\lambda^2}.\nonumber
\end{align}
Since $\Phi_k$ is symmetric, $\left\| \Phi_k \right\| = \rho(\Phi_k)$. By using the Gershgorin circle theorem, we have

\begin{align*}
	\begin{aligned}
		\left\|\Phi_{k}\right\| \leq & \max _{j} \sum_{i=1}^{k}\left[\Phi_{k}\right]_{i j} \\
		= & \max _{j}\left[\sum_{i=1}^{j}\left[\Phi_{k}\right]_{i j}+\sum_{i=j+1}^{k}\left[\Phi_{k}\right]_{j i}\right] \\
		= & \max _{j} \left[ \frac{\left(1-\lambda^{j}\right)\left(1-\lambda^{2(k-j+1)}\right)} {(1-\lambda)\left(1-\lambda^{2}\right)} \right. \\
		& \left.+\frac{\lambda\left(1-\lambda^{k+1-j}\right)+\lambda^{2(k-j+1)}\left(1-\lambda^{j-k-1}\right)}{(1-\lambda)\left(1-\lambda^{2}\right)}\right] \\
		\le & \frac{1}{(1-\lambda)^2}.
	\end{aligned}
\end{align*}
Thus we can obtain~\eqref{eq:bound_s4_matrix}.

\subsection{Proof of Theorem 1}\label{apdx:thm1}

Since $f_i$ is $L-$Lipschitz smooth, we have
\begin{align}\label{eq:L_smt_f}
	&f(\bar{\mathbf{x}}_{k+1})\nonumber \\
	\le &f(\bar{\mathbf{x}}_{k}) + \nabla f(\bar{\mathbf{x}}_{k})^{\top}(\bar{\mathbf{x}}_{k+1} - \bar{\mathbf{x}}_{k}) + \frac{L}{2}||\bar{\mathbf{x}}_{k+1}-\bar{\mathbf{x}}_{k}||^2 \nonumber \\
	= &f(\bar{\mathbf{x}}_{k}) - \alpha \nabla f(\bar{\mathbf{x}}_{k})^{\top} \bar{\mathbf{y}}_k + \frac{\alpha^2 L}{2}||\bar{\mathbf{y}}_k||^2.
\end{align}

For the term related to $\bar{\mathbf{y}}_k$, we have
\begin{align}\label{eq:y_bar}
	\bar{\mathbf{y}}_k 
	=& (\mathbf{Y}_k - \pi_{\mathbf{C}} \hat{\mathbf{y}}_k^{\top} )^{\top} \pi_{\mathbf{R}} + \nabla\mathbf{f} (\mathbf{X}_{k})^{\top}\mathbf{1} \pi_{\mathbf{C}}^{\top}\pi_{\mathbf{R}} \nonumber \\
	&+\nabla f (\bar{\mathbf{x}}_k)\pi_{\mathbf{C}}^{\top} \pi_{\mathbf{R}},
\end{align}
and hence,
\begin{align}\label{eq:y2_ub}
	\left \|\bar{\mathbf{y}}_k\right \|^2 \le & 3 \delta_{2,C}^2 \left \| \mathbf{Y}_k - \pi_{\mathbf{C}} \hat{\mathbf{y}}_k^{\top} \right \|_{\mathbf{C}}^2 \nonumber \\
	& + 3 L^2 n (\pi_{\mathbf{C}}^{\top} \pi_{\mathbf{R}})^2 \delta_{2,R}^2 \left \| \mathbf{X}_{k} - \mathbf{1} \bar{\mathbf{x}}_{k}^{\top} \right \|_{\mathbf{R}}^2 \nonumber \\
	& + 3 (\pi_{\mathbf{C}}^{\top} \pi_{\mathbf{R}})^2 \left \|\nabla f (\bar{\mathbf{x}}_k)\right \|^2.
\end{align}

It follows form~\eqref{eq:y_bar} that the term $-\nabla f(\bar{\mathbf{x}}_{k})^{\top} \bar{\mathbf{y}}_k$ in~\eqref{eq:L_smt_f} can be calculated as
\begin{align}\label{eq:bar_g_bar_y}
	& -\nabla f(\bar{\mathbf{x}}_{k})^{\top} \bar{\mathbf{y}}_k \nonumber \\
	\le	& -\left \|\nabla f (\bar{\mathbf{x}}_k)\right \|^2\pi_{\mathbf{C}}^{\top} \pi_{\mathbf{R}} +  \left \| \nabla f(\bar{\mathbf{x}}_k) \right \| \left \|\mathbf{Y}_k-\pi_{\mathbf{C}}\hat{\mathbf{y}}_k^{\top} \right \|_{\mathbf{C}} \delta_{2,C}  \nonumber \\
	& + \left \| \nabla f(\bar{\mathbf{x}}_k) \right \| \left \| \mathbf{X}_{k} - \mathbf{1} \bar{\mathbf{x}}_{k}^{\top} \right \|_{\mathbf{R}} L \sqrt{n} \pi_{\mathbf{C}}^{\top} \pi_{\mathbf{R}} \delta_{2,R}.	
\end{align}

Combining Equations~\eqref{eq:L_smt_f}--\eqref{eq:bar_g_bar_y}, we can derive that
\begin{align}\label{eq:f_update}
	f(\bar{\mathbf{x}}_{k+1}) 
	\le &  f(\bar{\mathbf{x}}_{k}) - \alpha \pi_{\mathbf{C}}^{\top} \pi_{\mathbf{R}} \left(1 - \frac{3\alpha L \pi_{\mathbf{C}}^{\top} \pi_{\mathbf{R}} }{2}\right)\left \|\nabla f (\bar{\mathbf{x}}_k)\right \|^2 \nonumber\\
	& + \alpha \delta_{2,C} \left \| \nabla f(\bar{\mathbf{x}}_k) \right \| \left \|\mathbf{Y}_k-\pi_{\mathbf{C}}\hat{\mathbf{y}}_k^{\top} \right \|_{\mathbf{C}}  \nonumber \\
	& + \alpha L \sqrt{n} \pi_{\mathbf{C}}^{\top} \pi_{\mathbf{R}} \delta_{2,R} \left \| \nabla f(\bar{\mathbf{x}}_k) \right \| \left \| \mathbf{X}_{k} - \mathbf{1} \bar{\mathbf{x}}_{k}^{\top} \right \|_{\mathbf{R}} \nonumber \\
	& + \frac{3\alpha^2 L}{2} \delta_{2,C}^2 \left \| \mathbf{Y}_k - \pi_{\mathbf{C}} \hat{\mathbf{y}}_k^{\top} \right \|_{\mathbf{C}}^2 \nonumber \\
	& + \frac{3\alpha^2 L^3}{2} n (\pi_{\mathbf{C}}^{\top} \pi_{\mathbf{R}})^2 \delta_{2,R}^2 \left \| \mathbf{X}_{k} - \mathbf{1} \bar{\mathbf{x}}_{k}^{\top} \right \|_{\mathbf{R}}^2.
\end{align}
Summing both sides of~\eqref{eq:f_update} over $0,\ldots,k-1$ and utilizing Lemma~\ref{lem6} and Lemma~\ref{lem7}, we have
\begin{align}\label{eq:induct_sum_f}
	& \alpha \pi_{\mathbf{C}}^{\top} \pi_{\mathbf{R}} \left(1 - \frac{3\alpha L \pi_{\mathbf{C}}^{\top} \pi_{\mathbf{R}} }{2}\right) \sum_{t=0}^{k-1}\left \|\nabla f (\bar{\mathbf{x}}_k)\right \|^2  \nonumber \\
	\le & f(\bar{\mathbf{x}}_{0})-f^* + \alpha \left[ \frac{ (c_0 b_1 + c_3b_2)}{1-\lambda} + \frac{(c_2 b_1+c_5b_2)e_0S_e}{2{\alpha}(1-\lambda)^2} \right]\nonumber\\
	& + \alpha^2 \left[ \frac{(c_0^2b_3 + c_3^2b_4)}{(1-\lambda)^2} + \frac{(c_2^2b_3+c_5^2b_4) e_0 S_e}{(1-\lambda)^2} \right] \nonumber\\
	& + \alpha^2 \left[ \frac{c_2 b_1+c_5b_2}{2} + \frac{c_1 b_1 + c_4b_2}{1-\lambda}\right]\sum_{t=0}^{k-1}\left \| \nabla f(\bar{\mathbf{x}}_t) \right \|^2 \nonumber\\
	& + \frac{\alpha^4 (c_1^2b_3 + c_4^2b_4)}{(1-\lambda)^2} \sum_{t=0}^{k-2}\left \| \nabla f(\bar{\mathbf{x}}_t) \right \|^2 \nonumber\\
	& + \alpha (c_0 b_1 + c_3b_2) \sum_{t=0}^{k-1}\lambda^{t} \left \| \nabla f (\bar{\mathbf{x}}_t) \right \|^2.	
\end{align}
If $\lambda^{\top} \le \frac{\alpha}{1-\lambda}$, i.e., $ t \ge k_0 \triangleq \frac{\ln{\alpha}-\ln{(1-\lambda)}}{\ln{\lambda}} $
It then follows that
\begin{align}
	& \sum_{t=0}^{k-1}\lambda^{t} \left \| \nabla f (\bar{\mathbf{x}}_t) \right \|^2 \nonumber \\
	= & \sum_{t=0}^{k_0} \lambda^{t} \left \| \nabla f (\bar{\mathbf{x}}_t) \right \|^2 + \sum_{t=k_0+1}^{k-1} \lambda^{t} \left \| \nabla f (\bar{\mathbf{x}}_t) \right \|^2 \nonumber \\
	\le & \sum_{t=0}^{k_0} \left \| \nabla f (\bar{\mathbf{x}}_t) \right \|^2 + \frac{\alpha}{1-\lambda} \sum_{t=0}^{k-1}\left \| \nabla f (\bar{\mathbf{x}}_t) \right \|^2. \nonumber \\
\end{align}
Thus for~\eqref{eq:induct_sum_f} we have
\begin{align}
	& \alpha \pi_{\mathbf{C}}^{\top} \pi_{\mathbf{R}} \left(1 - \frac{3\alpha L \pi_{\mathbf{C}}^{\top} \pi_{\mathbf{R}} }{2} \right) \sum_{t=0}^{k-1}\left \|\nabla f (\bar{\mathbf{x}}_k)\right \|^2 \nonumber \\
	\le & f(\bar{\mathbf{x}}_{0})-f^* + \frac{(c_2 b_1+c_5b_2)e_0S_e}{2(1-\lambda)^2} + \alpha \left[ \frac{ (c_0 b_1 + c_3b_2)}{1-\lambda} \right] \nonumber \\
	& + \alpha^2 \left[ \frac{c_0^2b_3 + c_3^2b_4+ (c_2^2b_3+c_5^2b_4) e_0 S_e}{(1-\lambda)^2}  \right] \nonumber \\
	& + \alpha^2 \left[\frac{ c_1 b_1 + c_4b_2+c_0 b_1 + c_3b_2}{1-\lambda} + \frac{c_1^2b_3 + c_4^2b_4}{(1-\lambda)^2} \right. \nonumber \\
	& \left.+ \frac{(c_2 b_1+c_5b_2)}{2}\right] \sum_{t=0}^{k-1}\left \| \nabla f (\bar{\mathbf{x}}_t) \right \|^2 \nonumber \\
	& + \alpha (c_0 b_1 + c_3b_2) \sum_{t=0}^{k_0} \left \| \nabla f (\bar{\mathbf{x}}_t) \right \|^2
\end{align}

Move all terms related to $\left \| \nabla f(\bar{\mathbf{x}}_t) \right \|^2$ to the left hand side and we have
\begin{align}\label{eq:sum_norm2_f_barx}
		& \gamma \sum_{t=0}^{k-1}\left \|\nabla f (\bar{\mathbf{x}}_t)\right \|^2  \nonumber\\
		\le & f(\bar{\mathbf{x}}_{0})-f^*  + \frac{(c_2 b_1+c_5b_2)e_0S_e}{2(1-\lambda)^2}\nonumber\\
		& +  \alpha  \frac{ (c_0 b_1 + c_3b_2)}{1-\lambda}\left(1+\sum_{t=0}^{k_0} \left \| \nabla f (\bar{\mathbf{x}}_t) \right \|^2\right) \nonumber\\
		& + \alpha^2 \left[ \frac{c_0^2b_3 + c_3^2b_4+ (c_2^2b_3+c_5^2b_4) e_0 S_e}{(1-\lambda)^2}  \right],	
\end{align}
where
\begin{align}
	\begin{aligned}
		\gamma = & \alpha  \pi_{\mathbf{C}}^{\top} \pi_{\mathbf{R}} -\alpha^2 \left[ \frac{3L (\pi_{\mathbf{C}}^{\top} \pi_{\mathbf{R}})^2 }{2} + \frac{ c_1 b_1 + c_4b_2+c_0 b_1 + c_3b_2}{1-\lambda} \nonumber \right. \nonumber \\
		& \left. + \frac{c_1^2b_3 + c_4^2b_4}{(1-\lambda)^2} + \frac{(c_2 b_1+c_5b_2)}{2}\right]> 0.
	\end{aligned}
\end{align}
Thus the step size $\alpha$ should satisfy
\begin{align}
	\alpha <& \pi_{\mathbf{C}}^{\top} \pi_{\mathbf{R}}\left[ \frac{3L (\pi_{\mathbf{C}}^{\top} \pi_{\mathbf{R}})^2 }{2} + \frac{ c_1 b_1 + c_4b_2+c_0 b_1 + c_3b_2}{1-\lambda} \nonumber \right. \nonumber \\
	& \left. + \frac{c_1^2b_3 + c_4^2b_4}{(1-\lambda)^2} + \frac{(c_2 b_1+c_5b_2)}{2}\right]^{-1}.
\end{align}
We can obtain~\eqref{eq:bound_s5} from~\eqref{eq:sum_norm2_f_barx} and~\eqref{eq:bound_s6} from Lemma~\ref{lem7}.
According to~\eqref{eq:bound_s5} and~\eqref{eq:bound_s6}, we have 
\begin{align*}
	\sum_{k=0}^{\infty} \left \|\nabla f (\bar{\mathbf{x}}_k)\right \|^2 < \infty,~\text{and}~ \sum_{t=0}^{\infty} \left \| \mathbf{X}_{t} - \mathbf{1} \bar{\mathbf{x}}_{t}^{\top} \right \|_{\mathbf{R}} ^2 < \infty.
\end{align*}
Hence, both the squared norm of consensus error $\left \| \mathbf{X}_{t} - \mathbf{1} \bar{\mathbf{x}}_{t}^{\top} \right \|_{\mathbf{R}} ^2$ and the squared gradient norm of average dual variable $\left \|\nabla f (\bar{\mathbf{x}}_k)\right \|^2$ converge to $0$ at an ergodic rate of $\mathcal{O}(1/k) $.

\subsection{Proof of Lemma 8}
To prove the linear convergence of $f$, we first establish the following linear inequality systems w.r.t. $\left \| \mathbf{X}_{k} - \mathbf{1} \bar{\mathbf{x}}_{k}^{\top} \right \|_{\mathbf{R}}^2$, $\left \| \mathbf{Y}_k - \pi_{\mathbf{C}} \hat{\mathbf{y}}_k^{\top} \right \|_{\mathbf{C}}^2$ and $f(\bar{\mathbf{x}}_k) - f^*$. For simplicity of notation, we denote $G_k = f(\bar{\mathbf{x}}_k) - f^*$.

\textbf{Step 1: The first inequality w.r.t. $\left \| \mathbf{X}_{k} - \mathbf{1} \bar{\mathbf{x}}_{k}^{\top} \right \|_{\mathbf{R}}^2$:}\\
Using~\eqref{eq:x_con_err_1} we obtain
\begin{align*}
	& \left\|  \mathbf{X}_{k+1} - \mathbf{1} \bar{\mathbf{x}}_{k+1}^{\top} \right \|_{\mathbf{R}}^2 \\
	\le & \left( \sigma_{\mathbf{R}}\left\| \mathbf{X}_{k} - \mathbf{1} \bar{\mathbf{x}}_{k}^{\top}\right\|_{\mathbf{R}} + \alpha\left\| \mathbf{Y}_k\right \|_{\mathbf{R}} \right)^2 + (1+\sigma_{\mathbf{R}})^2ne_k^2 \\
	&+ 2\left( \sigma_{\mathbf{R}}\left\| \mathbf{X}_{k} - \mathbf{1} \bar{\mathbf{x}}_{k}^{\top} \right\| + \alpha\left\| \mathbf{Y}_k \right \|_{\mathbf{R}}\right)(1+\sigma_{\mathbf{R}})\sqrt{n}e_k.
\end{align*}
Based on the Cauchy-Schwarz inequality, we have
\begin{align*}
	& \left( \sigma_{\mathbf{R}}\left\| \mathbf{X}_{k} - \mathbf{1} \bar{\mathbf{x}}_{k}^{\top}\right\|_{\mathbf{R}} + \alpha\left\| \mathbf{Y}_k\right \|_{\mathbf{R}} \right)^2 \\
	\le & \frac{2\sigma_{\mathbf{R}}^2}{1+\sigma_{\mathbf{R}}^2} \left\| \mathbf{X}_{k} - \mathbf{1} \bar{\mathbf{x}}_{k}^{\top}\right\|_{\mathbf{R}}^2 + \frac{2}{1-\sigma_{\mathbf{R}}^2} \alpha ^2 \left\| \mathbf{Y}_k\right \|_{\mathbf{R}}^2,
\end{align*}
and
\begin{align*}
	&2\left( \sigma_{\mathbf{R}}\left\| \mathbf{X}_{k} - \mathbf{1} \bar{\mathbf{x}}_{k}^{\top} \right\|_{\mathbf{R}} + \alpha\left\| \mathbf{Y}_k \right \|_{\mathbf{R}}\right)(1+\sigma_{\mathbf{R}})\sqrt{n}e_k \\
	\le & \frac{(1-\sigma_{\mathbf{R}}^2)\sigma_{\mathbf{R}}^2}{1+\sigma_{\mathbf{R}}^2} \left\|\mathbf{X}_{k} - \mathbf{1} \bar{\mathbf{x}}_{k}^{\top} \right\|_{\mathbf{R}}^2 + \alpha^2 \left\| \mathbf{Y}_k \right \|_{\mathbf{R}}^2 \\
	& + \frac{2(1+\sigma_{\mathbf{R}})^2}{1+\sigma_{\mathbf{R}}^2}ne_k^2.
\end{align*} 
For the terms regarding $\left \| \mathbf{Y}_k \right \|_{\mathbf{R}}^2$, we have
\begin{align*}
	\left \| \mathbf{Y}_k \right \|_{\mathbf{R}}^2
	& \le 2\delta_{\mathbf{R},\mathbf{C}}^2\left \| \mathbf{Y}_k - \pi_{\mathbf{C}} \hat{\mathbf{y}}_k^{\top} \right \|_{\mathbf{C}}^2 + 2\left \| \hat{\mathbf{y}}_k^{\top} \right \|_{\mathbf{R}}^2.
\end{align*}
And based on the $L$-Lipschitz smoothness of $f_i$, we have
\begin{align*}
	\left \| \hat{\mathbf{y}}_k^{\top} \right \|_{\mathbf{R}}^2 &= \left \| \mathbf{1}^{\top} \nabla\mathbf{f} (\mathbf{X}_{k})- \mathbf{1}^{\top} \nabla\mathbf{f} (\mathbf{1}\bar{\mathbf{x}}_k^{\top}) + \mathbf{1}^{\top} \nabla\mathbf{f} (\mathbf{1}\bar{\mathbf{x}}_k^{\top})\right \|_{\mathbf{R}}^2 \nonumber \\
	& \le 2  n L^2 \left\| \mathbf{X}_{k} - \mathbf{1} \bar{\mathbf{x}}_{k}^{\top} \right \|_{\mathbf{R}}^2 + 2  \left \| \nabla f(\bar{\mathbf{x}}_k) \right \|_2^2.
\end{align*}
Thus now we can obtain the inequality w.r.t. $\left\| \mathbf{X}_{k} - \mathbf{1} \bar{\mathbf{x}}_{k}^{\top} \right \|_{\mathbf{R}}^2$ in~\eqref{eq:inqeasys_PL}.

\textbf{Step 2: The second inequality w.r.t. $\left \| \mathbf{Y}_{k+1}-\pi_{\mathbf{C}}\hat{\mathbf{y}}_{k+1}^{\top} \right \|_{\mathbf{C}}^2$:}\\
From~\eqref{eq:Y_k+1}, we obtain that
\begin{align*}
	\begin{aligned}
		&\left\| \mathbf{Y}_{k+1} - \pi_{\mathbf{C}} \hat{\mathbf{y}}_{k+1}^{\top}  \right \|_{\mathbf{C}}^2  \\
		\le & \frac{1+\sigma_{\mathbf{C}}^2}{2} \left \| \mathbf{Y}_k - \pi_{\mathbf{C}} \hat{\mathbf{y}}_k^{\top}  \right \|_{\mathbf{C}}^2 + 2 \frac{1+\sigma_{\mathbf{C}}^2}{1-\sigma_{\mathbf{C}}^2} L^2 \left \| \mathbf{X}_{k+1}-\mathbf{X}_{k} \right \|_{\mathbf{C}}^2\\
		& + 2 \frac{1+\sigma_{\mathbf{C}}^2}{1-\sigma_{\mathbf{C}}^2} (1+ \sigma_{\mathbf{C}})^2 ne_k^2.
	\end{aligned}
\end{align*}

For $\left \| \mathbf{X}_{k+1}-\mathbf{X}_{k} \right \|_{\mathbf{C}}^2$, from~\eqref{eq:x_k+1 - x_k} we have 
\begin{align*}
	& \left \| \mathbf{X}_{k+1}-\mathbf{X}_{k} \right \|_{\mathbf{C}}^2  \\
	\le & \left( 4 \delta_{\mathbf{C},\mathbf{R}}^2 (1+\sigma_{\mathbf{R}})^2 + 8\alpha^2 \delta_{\mathbf{C},\mathbf{R}}^2nL^2 \right)\left\| \mathbf{X}_{k} - \mathbf{1} \bar{\mathbf{x}}_{k}^{\top} \right \|_{\mathbf{R}}^2 \\
	& + 4 \alpha^2 \left \| \mathbf{Y}_k - \pi_{\mathbf{C}} \hat{\mathbf{y}}_k^{\top}  \right \|_{\mathbf{C}}^2 + 8\alpha^2 \left \| \nabla f(\bar{\mathbf{x}}_k) \right \|_2^2\\
	&+ 4\delta_{\mathbf{C},\mathbf{R}}^2 (1+\sigma_{\mathbf{R}})^2ne_k^2.
\end{align*}
Thus we can obtain the inequality w.r.t. $\left\| \mathbf{Y}_{k} - \pi_{\mathbf{C}} \hat{\mathbf{y}}_{k}^{\top}  \right \|_{\mathbf{C}}^2 $ in~\eqref{eq:inqeasys_PL}.

\textbf{Step 3: The third inequality w.r.t. $f(\bar{\mathbf{x}}_k) - f^*$:}\\
By integrating P-Ł condition into~\eqref{eq:f_update} we have
\begin{align*}
	\begin{aligned}
		G_{k+1} \le & \left(1-2\mu\pi_{\mathbf{C}}^{\top} \pi_{\mathbf{R}}\alpha + 2L\alpha^2+3(\pi_{\mathbf{C}}^{\top} \pi_{\mathbf{R}})^2 L^2\alpha^2 \right )G_{k}\\ &+\frac{(1+3\alpha^2L)}{2}L^2 n (\pi_{\mathbf{C}}^{\top} \pi_{\mathbf{R}})^2 \delta_{2,R}^2 \left \| \mathbf{X}_{k} - \mathbf{1} \bar{\mathbf{x}}_{k}^{\top}\right \|_{\mathbf{R}}^2 \\
		& +  \frac{(1+3\alpha^2L)}{2} \delta_{2,C}^2 \left \| \mathbf{Y}_k - \pi_{\mathbf{C}} \hat{\mathbf{y}}_k^{\top}\right \|_{\mathbf{C}}^2.
	\end{aligned}
\end{align*}
Since $f_i$ is $L$-Lipschitz smooth for all $i \in \mathcal{N}$, the global dual objective $f$ is $nL$-Lipschitz smooth, i.e., $\forall \mathbf{x}, \mathbf{y} $ in $\mathbb{R}^m$,
\begin{align*}
	f(\mathbf{y}) \le f(\mathbf{x}) + \nabla f(\mathbf{x}) ^{\top}(\mathbf{y} - \mathbf{x}) + \frac{nL}{2} \left\| \mathbf{y} - \mathbf{x} \right \|^2.
\end{align*}
Substitute $\mathbf{y} = \mathbf{x} - \frac{1}{nL}\nabla f(\mathbf{x})$ into the equation, we obtain that
\begin{align*}
	f^* \le f(\mathbf{x} - \frac{1}{nL}\nabla f(\mathbf{x})) \le f(\mathbf{x}) - \frac{1}{2nL}\left \|\nabla f (\mathbf{x}_k)\right \|^2.
\end{align*}
Thus, by using the $nL$-Lipschitz smoothness of $f$, i.e., $\left \| f(\bar{\mathbf{x}}_k) \right \|^2 \le 2nLG_{k}$, we can obtain the inequality system~\eqref{eq:inqeasys_PL}.

\subsection{Proof of Theorem 2}
\label{apdx:thm2}
In light of Lemma~\ref{lem9}, we first provide conditions on $\alpha$ that ensure the spectral radius $\rho(\mathbf{P})<1$. Assume $\alpha < \frac{1}{2\beta\pi_{\mathbf{C}}^{\top} \pi_{\mathbf{R}}}$, it sufficies to ensure that diagonal entries $p_{11}, p_{22}, p_{33} <1$ and $\det(\mathbf{I}-\mathbf{P})>0$, or more aggressively,
\begin{align}\label{eq: det I-P}
	\begin{aligned}
		& \det(\mathbf{I}-\mathbf{P}) \\
		= &(1-p_{11})(1-p_{22})(1-p_{33})- p_{12}p_{23}p_{31}\\
		&- p_{13}p_{21}p_{32} - p_{13}(1-p_{22})p_{31} \\
		&- (1-p_{11} )p_{23}p_{32}- p_{12}p_{21}(1-p_{33})\\
		> & \frac{1}{2} (1-p_{11})(1-p_{22})(1-p_{33}).
	\end{aligned}
\end{align}
Now we consider sufficient conditions under wihch $p_{11}, p_{22}, p_{33}<1$ and~\eqref{eq: det I-P} holds.
First, $p_{11}, p_{22}<1$ is ensured respectively by
\begin{align*}
	1-p_{11} > \frac{(1-\sigma_{\mathbf{R}}^2)^2}{2(1+\sigma_{\mathbf{R}}^2)}, \quad 1-p_{22} > \frac{1-\sigma_{\mathbf{C}}^2}{4},
\end{align*}
which requires
\begin{align*}
	\alpha < \min \left \{ \frac{1}{2\beta\pi_{\mathbf{C}}^{\top} \pi_{\mathbf{R}}}, \frac{1-\sigma_{\mathbf{R}}^2}{4\sqrt{n}L\sqrt{1+\sigma_{\mathbf{R}}^2}},  \frac{1-\sigma_{\mathbf{C}}^2}{4\sqrt{2}L\sqrt{1+\sigma_{\mathbf{C}}^2}}\right \}.
\end{align*}
And $p_{33} < 1$ is ensured by
\begin{align*}
	\alpha < & \frac{2\beta \pi_{\mathbf{C}}^{\top} \pi_{\mathbf{R}}}{2L + 3L^2(\pi_{\mathbf{C}}^{\top} \pi_{\mathbf{R}})^2}.
\end{align*}

Then, notice that $1-p_{11} < \frac{(1-\sigma_{\mathbf{R}}^2)^2}{1+\sigma_{\mathbf{R}}^2}$ and $1-p_{22} < \frac{1-\sigma_{\mathbf{C}}^2}{2}$. Substituting the lower and upper bounds of $1-p_{11}$ and $1-p_{22}$ into~\eqref{eq: det I-P}, we obtain:
\begin{align*}\label{eq:lower bound det}
	&\frac{1}{2} (1-p_{11})(1-p_{22})(1-p_{33})- p_{12}p_{23}p_{31} - p_{13}p_{21}p_{32} \\ 
	&> h_1 \alpha -h_2 \alpha^2 -h_3 \alpha^3 -h_4 \alpha^4 -h_5 \alpha^5 -h_6\alpha^6 \\
	&> h_1 \alpha -h_2 \alpha^2 -(h_3+h_4+h_5+h_6)\alpha^3
\end{align*}
The final inequality is established by truncating the higher-order terms to $\alpha^3$. Its validity is ensured by the positivity of all coefficients and the assumption that $\alpha \in (0,1)$.
In order for~\eqref{eq: det I-P} to hold, it is sufficient that the lower bound in~\eqref{eq:lower bound det} be strictly positive:
\begin{align*}
	h_1\alpha - h_2\alpha^2 - (h_3+h_4+h_5+h_6) \alpha^3 > 0.
\end{align*}
Hence,
\begin{align*}
	\alpha <\frac{-h_2 + \sqrt{h_2^2 + 4h_1(h_3+h_4+h_5+h_6)}}{2(h_3+h_4+h_5+h_6)}.
\end{align*}
The following equation is established by induction to support the convergence analysis:
\begin{align}
	\mathbf{z}_k \le \mathbf{P}^k \mathbf{z}_{0} + \sum_{t=0}^{k-1}\mathbf{P}^{t} \mathbf{Q}e_{k-1-t}^2,
\end{align}
where $\mathbf{z}_k \triangleq [	\left \| \mathbf{X}_{k}-\mathbf{1}\bar{\mathbf{x}}_{k}^{\top} \right \|_{\mathbf{R}}^2, \left \| \mathbf{Y}_{k}-\pi_{\mathbf{C}}\hat{\mathbf{y}}_{k}^{\top} \right \|_{\mathbf{C}}^2, f(\bar{\mathbf{x}}_k) - f^*]^{\top}$. Assuming $e_t \le Es^{t}$, it can be obtained that
\begin{align}
	\left \| \mathbf{Q}e_t^2 \right \| \le h_7 E^2s^{2t}.
\end{align}

Since~\eqref{eq:PL_a_bd} provides sufficient condition on $\alpha$ that ensures $\lambda<1$, it holds that:
\begin{align}
	\left \| \mathbf{P}^{t}\right\| \le h_8 \rho^{t}(\mathbf{P}),
\end{align}
where $h_8 = \sqrt{3}\frac{\max_{1\le i\le 3}\nu_i}{\min_{1\le i\le 3}\nu_i}$ and $\nu = [\nu_1, \nu_2, \nu_3]^{\top}$ is an eigenvector of $\mathbf{P}$ associated with the eigenvalue $\lambda = \rho(\mathbf{P})$. If we choose the geometrically decaying error as $s\in (\sqrt{\lambda}, 1)$ and $E \ge \frac{\left\| \mathbf{z}_0 \right \| (s^2-\lambda)}{h_7}$, it holds that
\begin{align}
		\left\| \mathbf{z}_k \right \| \le & \left\| \mathbf{P}^k \mathbf{z}_{0} \right \| + \sum_{t=0}^{k-1}\left \|  \mathbf{P}^{t}\right\| 	\left \| \mathbf{Q}e_{k-1-t}^2 \right \|\nonumber \\
		\le & h_8\lambda^k\left\| \mathbf{z}_0 \right \| + \sum_{t=0}^{k-1} h_7h_8E^2\lambda^{t}s^{2(k-1-t)} \nonumber\\
		= & h_8\lambda^k \left(\left\| \mathbf{z}_0 \right \| -\frac{h_7E^2}{s^2 - \lambda} \right) + \frac{h_7h_8E^2 s^{2k}}{s^2 - \lambda}\nonumber \\
		\le & \frac{h_7h_8E^2 s^{2k}}{s^2 - \lambda}.
\end{align}
Therefore, we can directly obtain~\eqref{eq:PL_linear_converge} which concludes that both the consensus error $\left \| \mathbf{X}_{k}-\mathbf{1}\bar{\mathbf{x}}_{k}^{\top} \right \|_{\mathbf{R}}$ and the optimality gap $ f(\bar{\mathbf{x}}_k) - f^*$ converges to $0$ linearly.

\section*{References}
\bibliographystyle{IEEEtran}
\bibliography{main}

\end{document}